\documentclass[11pt]{article}

\usepackage[english]{babel}
\usepackage[utf8x]{inputenc}
\usepackage[T1]{fontenc}

\usepackage[margin=1.in,marginparwidth=1.75cm]{geometry}

\usepackage{complexity}
\usepackage[algo2e, ruled, vlined]{algorithm2e}
\usepackage{amsmath,amssymb}
\usepackage{amsfonts}
\usepackage{graphicx}
\usepackage{amsthm}
\usepackage{dsfont}
\usepackage{bbm}

\usepackage[colorinlistoftodos]{todonotes}
\usepackage[colorlinks=true, allcolors=blue]{hyperref}

\theoremstyle{theorem}
\newtheorem{theorem}{Theorem}
\newtheorem{claim}{Claim}
\newtheorem{fact}{Fact}
\newtheorem{lemma}{Lemma}
\newtheorem{corollary}{Corollary}

\newtheorem{remark}{Remark}
\newenvironment{reminder}[1]{\bigskip
	\noindent {\bf Reminder of #1.}\em}{\smallskip}

\theoremstyle{definition}
\newtheorem{definition}{Definition}

\newenvironment{proofof}[1]{\begin{proof}[{\textit{Proof of #1}}]}{\end{proof}}

\newcommand{\IT}{\mathbf{IT}}
\newcommand{\Ber}{\mathbf{Ber}}
\newcommand{\Bin}{\mathbf{Bin}}
\newcommand{\Lap}{\mathbf{Lap}}
\newcommand{\TLap}{\mathbf{TLap}}
\newcommand{\calA}{\mathcal{A}}
\newcommand{\calM}{\mathcal{M}}
\newcommand{\calB}{{\mathcal{B}}}

\newcommand{\calD}{\mathcal{D}}
\newcommand{\calE}{{\mathcal{E}}}
\newcommand{\calF}{\mathcal{F}}

\newcommand{\calX}{{\mathcal{X}}}
\newcommand{\calY}{{\mathcal{Y}}}
\newcommand{\Ex}{\mathbb{E}}
\newcommand{\eps}{\varepsilon}

\newcommand{\tamas}[1]{}
\newcommand{\edith}[1]{}
\newcommand{\xin}[1]{}
\newcommand{\uri}[1]{}

\newcommand{\ignore}[1]{}

\title{Generalized Private Selection and Testing with High Confidence}

\author{{\normalfont Edith Cohen}\thanks{Google Research and Tel Aviv University. \texttt{edith@cohenwang.com}.} \and Xin Lyu\thanks{UC Berkeley and Google Research. \texttt{lyuxin1999@gmail.com}. Most of the work was done while interning at Google in Summer 2022.} \and Jelani Nelson\thanks{UC Berkeley and Google Research. \texttt{minilek@alum.mit.edu}.} \and Tam\'{a}s Sarl\'{o}s\thanks{Google Research. \texttt{stamas@google.com}.} \and Uri Stemmer\thanks{Tel Aviv University and Google Research. \texttt{u@uri.co.il}.}}

\begin{document}

\maketitle

\begin{abstract}
Composition theorems are general and powerful tools that facilitate privacy accounting across multiple data accesses from per-access privacy bounds. However they often result in weaker bounds compared with end-to-end analysis.
Two popular tools that mitigate that are the exponential mechanism (or report noisy max) and the sparse vector technique. They were generalized in a couple of recent private selection/test frameworks, including the work by Liu and Talwar (STOC 2019), and Papernot and Steinke (ICLR 2022).

In this work, we first present an alternative framework for private selection and testing with a simpler privacy proof and equally-good utility guarantee. Second, we observe that the private selection framework (both previous ones and ours) can be applied to improve the accuracy/confidence trade-off for many fundamental privacy-preserving data-analysis tasks, including query releasing, top-$k$ selection, and stable selection.

Finally, for online settings, we apply the private testing to design a mechanism for adaptive query releasing, which improves the sample complexity dependence on the confidence parameter for the celebrated private multiplicative weights algorithm of Hardt and Rothblum (FOCS 2010). 
\end{abstract}

\section{Introduction}

Computations over large datasets often involve multiple intermediate accesses to the data by different algorithms with the final output being some function or aggregate of the outputs of these individual accesses.
When we are interested in differential privacy~\cite{DworkMNS16}, we typically have privacy or stability guarantees for each intermediate access which we would like to use to bound the end-to-end privacy cost.

One approach to do that is to use {\em composition} theorems~\cite{DBLP:conf/focs/DworkRV10}.  Compositions theorems yield overall privacy cost that scales linearly or sometimes even with square-root dependence in the number of accesses.  The disadvantage, however, is that compositions theorems are designed for a scenario where all the intermediate outputs are released and do not benefit in their privacy accounting from the 
particular way in which the results of the intermediate accesses are aggregated.
Powerful tools from the literature for mitigating this generally fall under the two categories of {\em private testing} and {\em private selection}.
Private testing includes 
the sparse vector technique~\cite{DNRRV:STOC2009,DBLP:conf/stoc/RothR10,DBLP:conf/focs/HardtR10}, where each access tests a hypothesis over the data, and only accesses with positive output incur a privacy ``charge.''  Private selection includes the exponential mechanism~\cite{DBLP:conf/focs/McSherryT07} that allows to select an approximate best solution (out of a large  set of possible solutions) 
while incurring a privacy cost close to that of estimating the quality of a single solution.  These tools have many variations and extensions~\cite{DBLP:journals/pvldb/LyuSL17, DBLP:conf/colt/KaplanMS21,DBLP:conf/iclr/Papernot022}, most prominently, a fairly recent framework by Liu and Talwar \cite{LiuT19-private-select}, and its subsequent work by Papernot and Steinke \cite{DBLP:conf/iclr/Papernot022}.

\SetKwFunction{FSelect}{Selection}
\SetKwFunction{FTest}{Test}
\SetKwFunction{FBest}{Best}
    
Following \cite{LiuT19-private-select,DBLP:conf/iclr/Papernot022}, we propose an alternative and simpler template for private selection \emph{and} testing that is described in Algorithm~\ref{algo:sync}\footnote{In the ITCS proceeding version of the paper, we described our template as more ``general'' and ``versatile'' than prior works. That was not accurate, as asymptotically our applications to data-analysis tasks (which will be introduced in Section~\ref{sec:appselect}) can be recovered by plugging in the prior framework directly. See also Section~\ref{sec:discussions} for discussions about the relations between our work and prior ones.}.  We view it as a system that provides two functions to users, $\FSelect$ and $\FTest$. The calls to these functions can be interleaved and adaptive (that is, inputs may depend on prior outputs).
Algorithm~\ref{algo:sync} is initialized by randomly sampling an internal parameter $p$ that later serves as a ``pass'' probability (where the data is not accessed if not ``passed''). The distribution of $p$ is controlled by the parameter $\gamma$, where a smaller $\gamma$ pushes the distribution closer to $0$, lowering the privacy cost, and a higher $\gamma$ pushes the distribution closer to $1$, increasing utility.  Generally, the choice of $\gamma$ provides a tradeoff between privacy and utility. We will illustrate the tradeoff in several applications later.

The function $\FSelect$ inputs a parameter $\tau$ and a collection of $k$ mechanisms $\{M_i\}_{i=1}^k$, all with outputs from the same ordered domain\footnote{For example, the output space of $M_i$ can consist of pairs $(s, q)\in X\times \mathbb{R}$, where $s$ is a solution and $q\in \mathbb{R}$ is a quality score measuring how good $s$ is. We can then order pairs in the decreasing order of $q$, and use $\FSelect$ to select a solution with the best quality score. We remark that this is also the illustrating example presented in \cite{LiuT19-private-select}.}. It generates a collection $S$ of outputs and then returns $\FBest(S):=\max_{x\in S} x$. 
The parameter $\tau$ may be set differently in each call to $\FSelect$, and its role is to amplify utility by repetition. Intuitively, a larger $\tau$ can compensate for a smaller $p$ as each of the mechanisms $M_i$ 
produces in expectation $\tau p$ candidates for the $\FBest$ selection.
Surprisingly, as we shall soon see, the privacy accounting allows us to gain from separating $p$ and $\tau$.
The function $\FTest$ provides a basic functionality of hypothesis testing subjected to the ``pass'' probability $p$. Utility amplification (that also compensates for small pass probability $p$) can be achieved by multiple calls to $\FTest$ with the same hypothesis.

\paragraph*{Meta Privacy Theorems.} 
We establish
meta privacy theorems for the system in 
Algorithm~\ref{algo:sync}.  We first provide separate bounds for $\FSelect$ and $\FTest$ calls.

\begin{theorem}[Privacy guarantee for $\FSelect$]\label{theo:meta-selection}
Suppose that following initialization of Algorithm~\ref{algo:sync}, there are $c\ge 1$ calls to $\FSelect$ (with possibly varying inputs $\tau, k$, and $\{M_i\}_{i=1}^k$). If the mechanisms fed into $\FSelect$ in this process satisfy $(\eps,\delta_i)$-DP with $\delta_i\ge 0$, then the list of $c$ outputs satisfies $((2c+\gamma)\eps, \tau \sum_{i}\delta_i)$-DP.
\end{theorem}

There has been previous works \cite{LiuT19-private-select,DBLP:conf/iclr/Papernot022} showing how one can do private selection while only incurring a constant factor increase in the privacy parameter. See Section~\ref{sec:discussions} for a survey of previous works and how our algorithm fits into the context.

\begin{theorem}[Privacy guarantee for $\FTest$]\label{theo:meta-test}
Suppose that following initialization of Algorithm~\ref{algo:sync}, there are repeated calls to \FTest until $c\ge 1$ $\top$ responses are received. If all the mechanisms fed into \FTest in this process satisfy $(\eps,\delta_i)$-DP with $\delta_i\ge 0$, then the collection of outputs  satisfies $((2c+\gamma)\eps, \sum_{i}\delta_i)$-DP or $\left(O\left(\eps \gamma + \eps\sqrt{c\log(1/\delta)}\right), \delta + \sum_{i}\delta_i\right)$-DP for all $\delta \in (0, 2^{-\sqrt{c}})$.
\end{theorem}

Liu and Talwar \cite{LiuT19-private-select} also studied the question of constructing a generalized $\FTest$ procedure. They showed that, when there is a perfect ``percentile oracle'' associated with the base algorithm $\calA$ (in our context, this is an oracle $p(D)$ such that $p(D) := \Pr[H(D) = 0]$), then one can design a simple $\FTest$ procedure with pure-DP guarantee. However, when such an oracle is not immediately available, \cite{LiuT19-private-select} designed a sophisticated method to estimate $p(D)$ by repeatedly running $H$. As a result, the algorithm and analysis get substantially more involved, and it only achieves approximate DP. Theorem~\ref{theo:meta-test} provides an algorithm that can achieve pure DP without assuming the percentile oracle.

The statements of Theorem~\ref{theo:meta-selection} and Theorem~\ref{theo:meta-test} provide separate bounds for the privacy cost of each type of call ($\FSelect$ or $\FTest$).  Importantly, our unified framework allows for interleaved calls where the inputs are adaptive (depending on prior outputs).  In this case, the privacy costs simply add up except for the $\gamma \eps$ factor:
If there are $c_1$ calls to $\FSelect$ and $c_2$ calls with $\top$ output to  $\FTest$, each with mechanisms that are
$(\eps,0)$-DP,  then the total privacy cost is $((2c_1 + 2c_2 +\gamma)\eps,0)$-DP.   

The proofs of the meta-privacy theorems are provided in Section~\ref{sec:metaprivacyproofs}.
Note that the $O(\eps \sqrt{c\log(1/\delta)})$ dependence in the statement of Theorem~\ref{theo:meta-test} does not follow as a direct consequence of advanced composition \cite{DBLP:conf/focs/DworkRV10}. It requires some delicate calculations and is one of the primary technical contributions of our work.

\begin{algorithm2e}[h]
\LinesNumbered
    \caption{Private Selection and Testing}
    \label{algo:sync}
    
    \SetKwProg{Fn}{Function}{:}{}
    \SetKwProg{Pg}{Program}{:}{\KwRet}
    
    \SetKwFunction{FSelect}{Selection}
    \SetKwFunction{FTest}{Test}
    
    \DontPrintSemicolon
    
    \SetKwFunction{FMain}{Main}
    \SetKwFunction{Ftest}{Test}
    
    \KwIn{
        A dataset $D = \{x_1,\dots, x_n\}$. Parameter $\gamma\in (0,\infty)$.
    }
    \SetKwProg{Init}{initialize}{:}{}
    \Init{}{
        Sample $p$ from $[0,1]$ where $\Pr[p \le x] = x^{\gamma}, \forall x\in [0,1]$.
    }
    
    \Fn(\tcp*[f]{$(M_i)$ are mechanisms with outputs from an ordered domain}){\FSelect{$\tau, k, M_1,\dots, M_k$}}{
        $S\gets \emptyset$\;
        \For{$i=1,\dots, k$}{
            \For{$j=1,\dots, \tau$}{
                $r_{i,j}\sim \Ber(p)$\;
                \If{$r_{i,j} = 1$}{
                    $s_{i,j}\gets M_i(D)$\;
                    $S \gets S \cup \{s_{i,j}\}$\;
                }
            }
        }
        \KwRet{$\FBest(S)$} \tcp*[f]{The top value in the ordered set $S$.} 
    }
    
    \Fn(\tcp*[f]{$H$ is a private algorithm with output $\{\top, \perp\}$}){\Ftest{$H$}}{
        $r\sim \Ber(p)$ \;
        \If{$r = 1$}{
            $\Gamma \gets H(D)$ \;  
            \KwRet{$\Gamma$}
        }
        \Else{
            \KwRet{$\perp$}
        }
    }
\end{algorithm2e}

\subsection{Relation to Prior Work}\label{sec:discussions}

Our framework and its applications can be viewed as both a natural extension of the sparse vector technique (SVT)~\cite{DNRRV:STOC2009,DBLP:conf/stoc/RothR10,DBLP:conf/focs/HardtR10} and an alternative to the private selection framework of
Liu and Talwar \cite{LiuT19-private-select}, and its subsequent work by Papernot and Steinke \cite{DBLP:conf/iclr/Papernot022}.

\paragraph{Comparison with SVT.} 
In a nutshell, SVT allows one to test hypotheses of the form ``$f_i(D) \lessapprox  t$'' (i.e., test if $f_i(D)$ is approximately below the threshold $t$ up to small additive error) while only incurring privacy loss for ``Above'' results. 
We first show that \FTest calls provide this functionality:
To test if $f_i(D) \lessapprox t$ we  specify an algorithm $H$ as follows: $H$ computes $\hat{f}_i(D) = f_i(D) + \Lap(1/\eps)$ and outputs $\top$ if and only if $\hat{f}_i(D) \ge t$. 
It is easy to see that $H$ is $(\eps,0)$-private. Given the pass parameter $p$, a single call of $\FTest(H)$ is insufficient for a reliable response and we therefore amplify the confidence by repetition. For example, when $p\sim [0,1]$ it holds that 
$\Pr[p\ge \beta] = 1-\beta$, therefore, if we 
repeat $\FTest(H)$ for $\frac{1}{\beta^2}$ times and none of these test returns $\top$, we are very confident that $f_i(D)\lessapprox t$. On the other hand, if $f_i(D) \le t - \frac{10}{\eps}\log(1/\beta)$, the probability that we observe a $\top$ response is at most $\beta$. 

Our meta-privacy proofs 
for \FSelect and \FTest turn out to be natural generalization of the proof of the standard SVT \cite{DBLP:journals/fttcs/DworkR14}. The role of the ``pass'' probability $p$ in our algorithm is very similar to the noisy threshold in SVT, and we prove the privacy property by carefully coupling the executions of Algorithm~\ref{algo:sync} on two neighboring inputs $D,D'$ with different $p$ parameters. See Section~\ref{sec:metaprivacyproofs} for the detail.

In Section~\ref{sec:apptest}, we further extend our framework to design a private \FTest procedure, which yields improved confidence-accuracy tradeoffs for SVT. In doing so, we borrow insights from several previous works. This includes the observation that SVT does not necessarily need to re-initialize the noisy threshold \cite{DBLP:journals/pvldb/LyuSL17}, the framework that uses R\'enyi-Differential Privacy \cite{DBLP:conf/csfw/Mironov17-renyi} to analyze SVT \cite{DBLP:conf/nips/ZhuW20}, and the reductions from analyzing general mechanisms to analyzing simpler coin-flipping games with binary outcomes \cite{DBLP:conf/soda/GuptaLMRT10, DBLP:journals/corr/OhV13, DBLP:journals/toc/MurtaghV18, DBLP:conf/colt/KaplanMS21}.

\paragraph{Comparison with prior private selection algorithms.} The hypothesis testing functionality $\Ftest$ improves and simplifies the previous algorithm in \cite{LiuT19-private-select}. It does not seem to have an analogue in \cite{DBLP:conf/iclr/Papernot022}. In the following, we compare our $\FSelect$ functionality with prior works.

The remarkable work of Liu and Talwar \cite{LiuT19-private-select} initiated the study of private selection from private candidates. They showed that, given an $(\eps,0)$-DP algorithm $\calA$, privately running $\calA$ for $\mathbf{T}$ rounds and selecting the best trial preserves $(3\eps,0)$-DP, as long as $\mathbf{T}$ is a random variable drawn from a Geometric distribution. They also showed that, to perform private selection with a reasonably good utility, one must suffer from a factor of $2$ in the privacy degradation. Liu and Talwar also gave a more involved version of the private selection algorithm with a factor of $(2+\alpha)$ blowup in privacy, where $\alpha\in (0,1)$ can be arbitrarily small (and the price is a larger running time).

The subsequent work by Papernot and Steinke \cite{DBLP:conf/iclr/Papernot022} presented a substantial improvement over the Liu-Talwar algorithm. They showed that, when $\mathbf{T}$ is a Poisson or a truncated negative binomial distribution, repeating the base algorithm $\calA$ for $\mathbf{T}$ times and selecting the best outcome preserve $((2+\alpha)\eps,0)$-DP where $\alpha$ can be arbitrarily small (notably, the algorithm is considerably simpler than the $((2+\alpha)\eps,0)$-DP algorithm of Liu and Talwar). They also extended the private selection framework to the R\'enyi-DP case, and conducted extensive experiments to compare different versions of private selection. 

Our result adds a class of new distributions to perform private selection. Namely, Theorem~\ref{theo:meta-selection} implies that, when $\mathbf{T}$ obeys a Beta-binomial distribution or a uniform distribution, we also obtain a private selection algorithm with $((2+\alpha)\eps,0)$-DP guarantee. In this work, we mainly consider the pure-DP setting. It is an interesting question to see if our result can generalize to the R\'enyi-DP setting. In terms of pure-DP, our algorithm offers similar trade-offs as the Papernot-Steinke algorithm. It is also interesting to evaluate and compare the practical performance of our variants of private selection.

\subsection{Applications}

We overview applications of our framework
in Section~\ref{sec:appLK}, where we demonstrate improvements of results in \cite{LiuT19-private-select}, and in Section~\ref{sec:appselect}, where we list direct applications of private selection to fundamental tasks in differential privacy, including
query releasing, top-$k$ selection and stable selection

A separate major contribution of our work is detailed in Section~\ref{sec:apptest}, where we apply our privacy accounting theorem (Theorem~\ref{theo:meta-test}) to improve the sparse vector techniqe (SVT) and through that, improve the private multiplicative weight update algorithm \cite{DBLP:conf/focs/HardtR10} and the sample complexity for adaptive data analysis with linear queries.

\subsection{Private Selection: Improving better-than-median selection} \label{sec:appLK}

Liu and Talwar~\cite{LiuT19-private-select} formulated the problem of finding a better-than-median solution of a private algorithm using oracle calls to the algorithm. The problem was also implicitly studied in \cite{DBLP:conf/soda/GuptaLMRT10}.

\begin{definition} [Private Better-than-median] \label{def:pBtM}
Let $\calA:\calX^n \to \calY\times \mathbb{R}$ be an $(\eps, \delta)$-DP algorithm, where the output of $\calA$ consists of a solution $y\in \calY$ and a score $s\in \mathbb{R}$. 

Fix $D\in \calX^n$ to be a dataset. Let $s_m = \mathrm{median}\{s(\calA(D))\}$ be the median score one gets when running $\calA$ on $D$.  A {\em better-than-median algorithm} $\calA^*$ with confidence parameter $\beta$ uses oracle calls to $\calA$ and is such that the score of the output of $\calA^*(D)$ is larger than $s_m$ with probability at least $1 - \beta$.
\end{definition}

We offer an algorithm that provides a smooth trade-off between utility and efficiency:
\begin{theorem}\label{theo:privacy-runtime-tradeoff}
For every $\alpha\in (0,\infty)$ and $\beta \in (0, 1)$, for $(\eps,\delta)$-DP $\mathcal{A}$ with $\eps\in (0,1)$ and $\delta\ge 0$, there is a better-than-median algorithm $\calA^*$ with confidence $1-\beta$ that satisfies the following: 
\begin{itemize}
\item With probability $1$, $\calA^*$ makes at most $T$ oracle calls to $\calA$, where
\[
T = \begin{cases}
~~\left\lceil \frac{2}{\beta} \right\rceil & \text{if $\alpha = 1$} \\
\left\lceil 5(\frac{2}{\beta})^{1/\alpha} \log(1/\beta)\right\rceil & \text{otherwise}
\end{cases}.
\]
\item $\calA^*$ is $((2+\alpha)\eps, T\delta)$-differentially private.
\end{itemize}
\end{theorem}

\begin{remark}
We note that one can obtain similar (i.e., asymptotically matching) results as Theorem~\ref{theo:privacy-runtime-tradeoff} by using the Liu-Talwar algorithm or the Papernot-Steinke algorithm. Since this work mainly focuses on asymptotic behaviors of privacy-preserving data-analysis tasks, we did not attempt to compare the hidden constants in running-time/utility/privacy with the different versions of instantiation of Theorem~\ref{theo:privacy-runtime-tradeoff}.
\end{remark}

Theorem~\ref{theo:privacy-runtime-tradeoff} follows using Algorithm~\ref{algo:sync}:
Consider the case $\alpha=1$ and (arbitrary) $\beta \in (0,1)$. $\calA^*$ initializes 
Algorithm~\ref{algo:sync} with $\gamma = 1$ and
calls \FSelect with arguments $(k = 1, \tau = \lceil\frac{2}{\beta} \rceil, M_1 = \calA)$. 
It follows from Theorem~\ref{theo:meta-selection} that $\calA^*$ is $(3\eps,\tau \delta)$-DP. As for the utility, note that $\FSelect$ runs $\calA$ for $m\sim \Bin(\tau, p)$ times where $p$ is uniformly random in $[0,1]$. Equivalently, $m$ is uniformly random in $\{0,1,\dots, \tau\}$\footnote{As for all $m\in\{0,1,\ldots,\tau\}, \Ex_{p\sim U[0,1]} {\tau \choose m} p^m(1-p)^{\tau-m} = 1/(\tau+1)$.}. Therefore, the probability that $\calA^*$ fails to get a better-than-median solution is at most
\[
\frac{1}{\tau + 1} \sum_{m=0}^{\tau} 2^{-m} \le \frac{2}{\tau + 1} \le \beta.
\]

The privacy-computation tradeoffs by varying $\alpha$, as stated in Theorem~\ref{theo:privacy-runtime-tradeoff}, are obtained by varying $\gamma$ in the initialization of Algorithm~\ref{algo:sync}. We provide details in Section~\ref{sec:proof-LT}. 



\subsection{Direct Applications of Private Selection} \label{sec:appselect}

In the following, we mention several applications of the private selection algorithm. We note that the applications in this subsection can also be obtained by using the Liu-Talwar or the Papernot-Steinke algorithm directly (i.e., all of the three algorithms give asymptotically the same utility/privacy trade-offs). The primary contribution in this subsection is making the connection between the private selection framework and these concrete applications.

\subsubsection{Query Releasing} 

Consider the task of answering $k$ sensitivity-$1$ queries $f_1,\dots, f_k: \calX^n \to \mathbb{R}$. A fundamental and extensively studied~\cite{SteU17-selection-lb, GaneshZ2021, DBLP:conf/colt/Ghazi0M21, DBLP:journals/corr/DaganK20-bounded}
problem is to characterize the privacy-accuracy tradeoff for the query releasing task. That is, for a given privacy budget $(\eps,\delta)$, determine the smallest error $e$, such that there is an $(\eps,\delta)$-DP algorithm that releases answers to all of the $k$ queries within $\ell_{\infty}$-error $e$. 
A naive application of the Gaussian mechanism gives expected $\ell_\infty$-error $O\left(\frac{\sqrt{k\log(1/\delta)\log(k)}}{\eps} \right)$, which was improved to $O\left(\frac{\sqrt{k\log(1/\delta)\log\log(k)}}{\eps} \right)$ by Steinke and Ullman \cite{DBLP:journals/jpc/SteinkeU16}. Recently, the question was raised again as an open question \cite{DPorg-open-problem-avoid-union} and two significant advances were made \cite{DBLP:conf/colt/Ghazi0M21,DBLP:journals/corr/DaganK20-bounded}. We resolve the problem completely by applying the private selection framework together with an algorithm from \cite{DBLP:conf/colt/Ghazi0M21}:

\begin{theorem}\label{theo:release-k-queries-non-adaptive}
There is a constant $C > 0$ such that the following holds. For every $\eps \in (0,1),\delta \in (0,1/2)$ and $k\in \mathbb{N}$, there is an $(\eps, \delta)$-DP algorithm that answers $k$ given sensitivity-$1$ queries $f_1,\dots, f_k$, such that
\[
\Pr_{\tilde{f}}\left[\| f - \tilde{f} \|_{\infty} \le C \frac{\sqrt{k\log 1/\delta}  + \log(1/\delta)}{\eps} \right] = 1,
\]
where $\tilde{f} := (\tilde{f}_1,\dots, \tilde{f}_k)$ denotes the responses released by the algorithm.
\end{theorem}

\begin{remark}
The $\frac{\log(1/\delta)}{\eps}$ term in the error is necessary if one desires the answers to be accurate with probability one. There is a lower bound even for the case $k = 1$\footnote{To see this lower bound, let $\calA$ be an arbitrary query-releasing mechanism. We construct a sensitivity-$1$ function $f$ and two private inputs $D^0,D^1$ of distance $\frac{\log(1/\delta)}{5\eps}$ such that $f(D^0)-f(D^1) = \frac{\log(1/\delta)}{5\eps}$. Since $\calA$ is $(\eps,\delta)$-DP, the supports of $\calA(f,D^0)$ and $\calA(f,D^1)$ overlap, which means that there is $v\in \mathbb{R}$ and $D\in\{D^0,D^1\}$ such that $|v-f(D)| \ge \frac{\log(1/\delta)}{10\eps}$ but $\Pr[\calA(f,D)=v] > 0$. }.
\end{remark}

Theorem~\ref{theo:release-k-queries-non-adaptive} is tight, and improves on the two incomparable previous works \cite{DBLP:conf/colt/Ghazi0M21,DBLP:journals/corr/DaganK20-bounded}.

\begin{itemize}
    \item \cite{DBLP:journals/corr/DaganK20-bounded} shows how to achieve the aforementioned accuracy with probability $1$. However, their algorithm only works for $\delta > 2^{-\frac{k}{\log^2 k \log^4(\log k)}}$.
    \item \cite{DBLP:conf/colt/Ghazi0M21} shows an algorithm for the full range of $\delta \in (0,1/2)$. However, their algorithm only promises to release an answer vector $\tilde{f}$ that is accurate (in $\ell_{\infty}$ sense) with probability $1-\frac{1}{\mathrm{poly}(k)}$. It was posed as an open question in \cite{DBLP:conf/colt/Ghazi0M21} whether one can release accurate responses (within $\ell_{\infty}$ distance) with probability one. Theorem~\ref{theo:release-k-queries-non-adaptive} answers this question affirmatively.
\end{itemize}

\paragraph*{Proof intuition.} In \cite{DBLP:conf/colt/Ghazi0M21}, it was shown how to $(\eps,\delta^2)$-privately release a vector $(\tilde{f}_1,\dots, \tilde{f}_k)$ such that, with probability at least $\frac{1}{2}$, one has
\[
\| f - \tilde{f} \|_{\infty} \le O\left( \frac{\sqrt{k\log(1/\delta)}}{\eps} \right).
\]
In the following, we use $\calA$ to denote the GKM algorithm with privacy parameters $(\eps,\delta^2)$. 

The high level idea is to use the better-than-median algorithm of Theorem~\ref{theo:privacy-runtime-tradeoff} to amplify the success probability. Consider applying Theorem~\ref{theo:privacy-runtime-tradeoff} on top of $\calA$ with $\alpha = 1$ and $\beta = \delta$. Theorem~\ref{theo:privacy-runtime-tradeoff} shows that the number of oracle calls to $\calA$ is bounded by $O(1/\delta)$. Then, it follows that the resulting algorithm is $(O(\eps),O(\delta))$-differentially private. Moreover, the algorithm releases an accurate vector with probability $1-\delta$. We can further improve the success probability to $1$ by non-privately correcting erroneous vectors. Since we do the non-private correction with probability at most $\delta$, the result algorithm is still $(O(\eps),O(\delta))$-DP. See Section~\ref{sec:proof-query-release} for the detail.

\subsubsection{Top-$k$ Selection}

We consider optimization over a finite set, which is another fundamental task in DP. Suppose there are $m$ candidates. Each of them is associated with a sensitivity-$1$ score function $f_1,\dots, f_m : \mathcal{X}^n\to \mathbb{R}$. We want to select $k$ candidates from the $m$ options with largest scores. This is a well-known and fundamental task that has been studied extensively (see, e.g., \cite{DBLP:conf/focs/McSherryT07, DBLP:journals/fttcs/DworkR14, McS20-permute, DBLP:conf/nips/DurfeeR19, DongDR20-exp, DBLP:conf/icml/QiaoSZ21}).

To measure the utility of a selection algorithm, a common criterion is the suboptimality gap. Suppose the algorithm releases a subset $S\subseteq [m]$ of candidates. The suboptimality gap of $S$ is defined as
\[
\Gap(S) := \max_{i\not\in S} \{ f_i(D) \} - \min_{j\in S} \{ f_j(D) \}.
\]
Namely, this is the difference between the \emph{best} option outside $S$ and the \emph{worst} option inside $S$.

The standard algorithm for top-$k$ applies the exponential mechanism sequentially for $k$ times. If we aim for $(\eps,\delta)$-DP for the final algorithm, each round of the exponential mechanism needs to be $(\frac{\eps}{\sqrt{k\log(1/\delta)}}, 0)$-DP. Let $\mathbf{S}$ denote the output of this algorithm. The standard calculation shows that
\[
\Pr\left[ \Gap(\mathbf{S})\ge \frac{\sqrt{k\log(1/\delta)}}{\eps} \cdot \log(m/\beta) \right] \le \beta.
\]
Namely, with probability $1-\beta$, the suboptimality gap is $\frac{\sqrt{k\log(1/\delta)}}{\eps} \log(m/\beta)$. If one is interested in very high confidence $\beta \le m^{-\omega(1)}$, this bound on the gap can be very large. We offer an algorithm where the suboptimality gap has improved dependency on $1/\beta$.

\begin{theorem}\label{theo:top-k-selection}
There is an absolute constant $C>0$ such that the following holds. For every $\eps,\delta,\beta\in (0, 1)$, there is an $(\eps,\delta)$-DP algorithm for top-$k$ selection satisfying the following. Let $\mathbf{S}$ denote the output of the algorithm. Then
\[
\Pr\left[ \Gap(\mathbf{S})\ge C\frac{\sqrt{k\log(1/\delta)}}{\eps} \cdot \log(m) + C\frac{\log(1/\beta)}{\varepsilon} \right] \le \beta.
\]
\end{theorem}

While we cannot avoid paying $\log(1/\beta)$ completely, we only pay $\log(1/\beta)$ \emph{once} instead of $\sqrt{k\log(1/\delta)}$ times as was in the standard algorithm.

The proof of Theorem~\ref{theo:top-k-selection} uses a similar strategy as that of Theorem~\ref{theo:release-k-queries-non-adaptive}. Suppose we iteratively run the exponential mechanism for $k$ rounds to get $k$ candidates $\mathbf{S}$. With probability at least $\frac{1}{2}$ we have that $\Gap(\mathbf{S}) \le O(\frac{\sqrt{k\log(1/\delta)}}{\eps} \cdot \log(m))$. We can use Theorem~\ref{theo:privacy-runtime-tradeoff} on top of the naive algorithm to boost the success probability from $\frac{1}{2}$ to $1-\beta$. See Section~\ref{sec:proof-top-k}.

\subsubsection{Stable Selection}\label{sec:stable-selection}

The private selection algorithm also has implications to stable selections. Generally, there is a lower bound saying that the $\Omega(\frac{1}{\eps}\log(m))$ suboptimality gap is necessary when we select from $m$ candidates with sensitivity-$1$ score functions. However, when there is a ``structure'' within the candidate family, we can often do better. We show two representative examples in the following.

\paragraph*{Choosing Mechanism.} Let $\mathcal{F}$ be a family of sensitivity-$1$ functions. Say that $\calF$ is $k$-bounded, if for any two neighboring datasets $D,D'$, it holds that $\sum_{f\in \calF}|f(D)-f(D')|\le k$. For $k$-bounded function families, we can do private selection with suboptimaility gap $O(\frac{1}{\eps}\log(k))$.


\begin{theorem}\label{theo:choosing-mechanism}
Suppose $\calF$ is a $k$-bounded function family. For every $\eps,\delta\in (0, 1)$, there is an $(\eps,\delta)$-DP algorithm $\calA$ that holds a dataset $D$ and selects a function $f\in \calF$ such that, with probability at least $1-\beta$,
\[
f(D) \ge \max_{f^*\in \calF}\{ f^*(D)\} - O(\frac{1}{\eps} \log(k/\delta\beta)).
\]
\end{theorem}

Theorem~\ref{theo:choosing-mechanism} generalizes several previous mechanisms for this setting \cite{DBLP:journals/toc/BeimelNS16,DBLP:conf/nips/Ghazi0M20}, which all require additional assumptions on the family $\calF$.

\paragraph*{Stable Selection.} In some settings, there is a noticeable quality gap between the best solution in $\calF$ and the $(k+1)$-th best solution. In this case, we can also obtain a private selection algorithm with improved suboptimality gap. In more detail, let $\calF$ be the function family. For any given dataset $D$, sort functions in $\calF$ by the decreasing order of $f(S)$. Namely, write $\calF = \{f_1,f_2,\dots, f_{m}\}$ so that $f_1(D) \ge f_2(D) \ge \dots \ge f_m(D)$. Define the $k$-th gap of $S$ on $\calF$ as $G_k(S,\calF) = f_1(D) - f_{k+1}(D)$.

\begin{theorem}\label{theo:stable-selection}
For any $k\ge 1$ and $\eps,\delta, \beta \in (0, 1)$, there is an $(\eps,\delta)$-DP algorithm $\calA$ that holds a dataset $D$ and achieves the following. For any $1$-Lipschitz function family $\calF$, suppose $G_k(D,\calF)\ge \frac{10}{\eps}\log(k/\delta\beta)$. Then, with probability at least $1-\beta$, $\calA$ selects a function $f\in \calF$ such that
\[
f(D) \ge \max_{f^*\in \calF}\{ f^*(D)\} - O(\frac{1}{\eps}\log(k/\delta\beta)).
\]
\end{theorem}

There is a known algorithm for this task with asymptotically the same performance (see \cite{DBLP:conf/stoc/BunDRS18, DBLP:journals/tit/BunKSW21} and references therein). However, our algorithm and its analysis are very different. In particular, we demonstrate how easily this problem can be solved under the private selection framework.

The proofs of Theorems~\ref{theo:choosing-mechanism} and \ref{theo:stable-selection} are deferred to Section~\ref{sec:proof-stable-selection}. 

\subsection{Improved Sparse Vector Technique with Applications} \label{sec:apptest}

In Section~\ref{sec:appselect}, we have shown that one can use the private selection framework to improve the confidence-accuracy tradeoff for query releasing and top-$k$ selection. However, note that these two applications work in the \emph{offline} setting: in the query releasing and top-$k$ selection tasks, we are given the $k$ queries and the $m$ candidates in advance, respectively. 

It is natural to ask if similar improvements can be obtained in the online setting. In the following, we demonstrate this is possible by designing a modification of the well-known sparse vector technique, which offers an improved confidence-accuracy tradeoff. We further apply our SVT to improve the private multiplicative weight algorithm \cite{DBLP:conf/focs/HardtR10} and the sample complexity of adaptive data analysis.

\subsubsection{Review of Sparse Vector Technique and its Applications}

Roughly speaking, the sparse vector technique provides an algorithm to answer an \emph{unbounded number} of queries of the form $f_i(S)\lessapprox t$, while one only incurs privacy loss for queries that are above the thresholds. See Section~\ref{sec:svt-improve-techinque} for a more technical and precise description.

\paragraph*{Application to Private Multiplicative Weights.} SVT is the main subroutine in the celebrated private multiplicative weights algorithm \cite{DBLP:conf/focs/HardtR10}, which can be used to answer a large number of linear queries subject to privacy constraints. We recall the problem setup first. Suppose there is a universe $\calX$ of size $|\calX|$. We want to design an algorithm $\calA$ that receives a private input $S=(x_1,\dots,x_n)\in \calX^n$ consisting of $n$ data points from $\calX$, and answers $m$ queries from an adversary $\calB$. Each query is specified by a function $f_i:\calX^n\to [0,1]$, and the algorithm needs to report an estimation $\tilde{f}_i$ for $\Ex_{j\sim [n]}[f_i(x_j)]$. The error of the algorithm is defined as $\max_{i\in [m]}\{ |\tilde{f}_i - \Ex_{j\sim[n]}[f_i(x_j)] |\}$.

The private multiplicative weight update (MWU) algorithm allows one to privately answer $m$ linear queries, while the error grows only logarithmically with $m$, as shown in the following theorem.

\begin{theorem}[\cite{DBLP:conf/focs/HardtR10}, see also \cite{ DBLP:journals/fttcs/DworkR14}]\label{theo:MWU}
For every $\eps,\delta\in (0,1)$, there is an $(\eps,\delta)$-DP algorithm that answers $m$ adaptively chosen linear queries with the following utility guarantee. Let $(\tilde{f}_i)_{i\in [m]}$ denote the outputs of $\calA$. For every $\beta \in (0,1)$, with probability $1-\beta$ we have that $\max_{i\in [m]} \{ |\tilde{f}_i - \Ex_{j\in [n]}[f_i(x_j)]|\} \le\alpha = \alpha(\beta)$, where
\[
\alpha(\beta)^2 \le O\left( \frac{\sqrt{\log|\calX|\log(1/\delta)} \log(m/\beta)}{n\eps} \right).
\]
\end{theorem}

Theorem~\ref{theo:MWU} can be used to perform adaptive data analysis with linear queries. We recall the setup: given the universe $\calX$, there is a distribution $\calD$ over $\calX$. An adaptive data analysis algorithm receives as input $n$ samples from $\calD$, denoted by $S\sim \calD^n$. It needs to adaptively answer $m$ linear queries, where each query is specified by a function $f_i:\calX\to [0,1]$ and the algorithm needs to output an estimation of $\Ex_{x\sim \calD}[f(x)]$. An algorithm is called $(\alpha, \beta, m)$-accurate, if it can answer $m$ (adaptive) linear queries within additive error $\alpha$ with probability at least $1-\beta$. We consider the \emph{sample complexity} of the problem. That is, fixing $\calX, m, \alpha, \beta$, we want to find out the smallest $n$ such that there exists an $(\alpha,\beta,m)$-accurate algorithm that only uses $n$ samples from $\calD$.

The generalization property of DP \cite{DFHPTRR15,DBLP:conf/stoc/BassilyNSSSU16} shows that, if we design an algorithm to output estimations of $\Ex_{x\sim S}[f_i(x)]$ in a privacy-preserving manner, then $\Ex_{x\sim S}[f_i(x)]$ is a good approximation of $\Ex_{x\sim \calD}[f_i(x)]$. Combining the privacy-preserving algorithm by Theorem~\ref{theo:MWU} with the generalization property, one obtains the following well-known sample complexity upper bound.

\begin{theorem}[\cite{DBLP:conf/stoc/BassilyNSSSU16}]\label{theo:adaptive-analysis-known}
For every finite universe $\calX$, every $m\ge 1$ and $\alpha,\beta \in (0, 1)$, there is an $(\alpha, \beta, m)$-accurate adaptive data analysis algorithm with sample complexity
\[
n  \le O\left( \frac{\sqrt{\log |\calX|\log(1/\beta)} \log(m/\beta) }{\alpha^3} \right).
\]
\end{theorem}

\paragraph*{Our improvement.} Using our improved sparse vector algorithm, we obtain the following improved private multiplicative weight update algorithm.

\begin{theorem}\label{theo:MWU-new}
There is an absolute constant $C>0$ such that the following is true. For every $\eps,\delta\in (0,1)$, $m\in \mathbb{N}$ and $\beta \ge 2^{-m}$, there exists 
\[
\alpha \le C \left( \frac{\sqrt{\log|\calX|\log(1/\delta)} \log(m)}{n\eps\cdot \alpha} + \frac{\log(1/\beta)}{n\eps} \right)
\]
and an $(\eps,\delta)$-DP algorithm that answers $m$ adaptively chosen linear queries such that, with probability $1-\beta$, we have $\max_{i\in [m]} \{ |\tilde{f}_i - \Ex_{j\in [n]}[f_i(x_j)]|\}\le \alpha$.
\end{theorem}

We get the following improved sample complexity for adaptive data analysis as a corollary of Theorem~\ref{theo:MWU-new} and the generalization property by \cite{DBLP:conf/stoc/BassilyNSSSU16}.

\begin{theorem}\label{theo:adaptive-analysis-improved}
For every finite universe $\calX$, every $m\ge 1$ and $\alpha\in (0, 1), \beta \ge 2^{-m}$, there is an $(\alpha, \beta, m)$-accurate adaptive data analysis algorithm with sample complexity
\[
n \le O\left( \min\left( \frac{\sqrt{\log|\calX|\log (1/\beta)}}{\alpha^3}, \frac{\log|\calX|}{\alpha^4} \right) \log m + \frac{\log(1/\beta)}{\alpha^2}\right).
\]
\end{theorem}

We compare Theorem~\ref{theo:adaptive-analysis-improved} with Theorem~\ref{theo:adaptive-analysis-known}. Overall, the sample complexity given by Theorem~\ref{theo:adaptive-analysis-improved} scales at most linearly with $\log(1/\beta)$, and we manage to ``decouple'' the dependence on $\log(1/\beta)$ from that on other parameters. For the non-adaptive setting, it is known (by the Chernoff bound) that the optimal sample complexity is $\Theta(\frac{\log m/\beta}{\alpha^2})$. Therefore, in the high confidence regime (in particular, when $\log(1/\beta) > \frac{\log|\calX|}{\alpha^2} \log m$), the bound given by Theorem~\ref{theo:adaptive-analysis-improved} asymptotically matches the non-adaptive bound.

\subsubsection{Intuition}\label{sec:svt-improve-techinque}

Now we discuss the proof idea of our improved sparse vector technique. 

\paragraph*{Review of the standard SVT.} Let us briefly review how the standard SVT works. Let $X$ be the private input, $t\in \mathbb{R}$ be a threshold, and $f_1,\dots, f_m$ be a list of sensitivity-$1$ queries. The sparse vector algorithm works by adding independent Laplace noises $\Lap(1/\eps)$ to the threshold and each query function. Then, it outputs the first index $i$ such that the (noisy) query value $\hat{f}_i(X)$ is larger than the (noisy) threshold $\hat{t}$.

The standard approach to argue the utility of the algorithm works as follows. First, with probability at least $1-\beta$, all the Laplace noises sampled in the SVT algorithm are bounded by $\frac{1}{\varepsilon}\log(m/\beta)$. Conditioning on this event and letting $i^*$ be the index returned by the algorithm, we have that $f_{i^*}(X) \ge t - \frac{2}{\eps}\log(m/\beta)$ and $f_{j}(X)\le t + \frac{2}{\eps}\log(m/\beta)$ for every $j < i^*$. Therefore, we conclude that the algorithm is ``approximately accurate'' within error $O(\frac{1}{\eps}\log(m/\beta))$ with probability $1-\beta$.

In many scenarios (including the private MWU algorithm), one usually wants to repeatedly run SVT for $k$ rounds, to identify $k$ meaningful queries. Suppose we desire the final algorithm to be $(\eps,\delta)$-private. In each round of SVT, we need to set the privacy budget as $\eps' = \frac{\eps}{\sqrt{k\log(1/\delta)}}$. Doing the analysis above, the error bound becomes $\frac{1}{\eps}\sqrt{k\log(1/\delta)}\log(m/\beta)$. This is somewhat unsatisfactory, as we need to pay the $\log(1/\beta)$ term ``$O(\sqrt{k})$'' times\footnote{However, we remark that the $\sqrt{k}\log(m)$ term is necessary (i.e., there is a lower bound of $\Omega(\sqrt{k}\log(m))$. See, e.g. \cite{BUV14-lb, SteU17-selection-lb, DBLP:books/sp/17/Vadhan17-dp-complex}).}.

\paragraph{Our improvement.} In the following, we show a modification of the sparse vector algorithm with better confidence-accuracy trade-off. In particular, our algorithm is accurate within error $O(\frac{1}{\eps}\sqrt{k\log(1/\delta)}\log(m) + \frac{\log(1/\beta)}{\eps})$ with probability $1-\beta$. To compare, the standard SVT has error bound $O(\frac{1}{\eps}\sqrt{k\log(1/\delta)}\log(m
/\beta))$. See Algorithm~\ref{algo:repetitive-SVT} in Section~\ref{sec:proof-SVT} for the complete description of our new algorithm; we explain the intuition and the main technical challenges below.

Recall that we want to run SVT for $k$ rounds, and desire the final privacy guarantee to be $(\eps,\delta)$-DP. Given this constraint, each round of SVT has to be $(\eps',0)$-DP for $\eps' = \frac{\eps}{\sqrt{k\log(1/\delta)}}$. 

To illustrate the idea, suppose for now that we do not need to add a noise to the threshold. (Namely, let us suppose that $\hat{t} = t$ always holds.) Now, given a query $f_i(X)$, the sparse vector algorithm first obtains a noisy estimation $\hat{f}_i(X) = f_i(X) + \Lap(1/\eps')$, and then compares $\hat{f}_i(X)$ against $\hat{t}$. If the comparison result is $\hat{f}_i(X) \le \hat{t}$, then we are $(1-\beta)$-confident that $\hat{f}_i(X) \le \hat{t} + \frac{2}{\eps'}\log(1/\beta)$. Observe that the error gap $\frac{2}{\eps'}\log(1/\beta)$ keeps increasing as we desire a higher and higher confidence.

We can amplify confidence more economically by repetition. Let $\tau = \log(1/\beta)$. Consider drawing $\tau$ independent noisy estimations of $f_i(X)$. Namely, we calculate $f_i^{(j)}(X) = f_i(X) + \Lap(1/\eps')$ for every $j\in [\tau]$. Suppose that all of these noisy estimations are below $\hat{t}$. Then we are $(1-\beta)$-confident that $f_i(X)\le \hat{t}$. Here, somewhat magically, we get the higher confidence without compromising privacy, as the SVT algorithm only incurs a privacy loss when one observes an ``Above-Threshold'' result. 

However, the situation becomes more complicated when we get the ``Above-Threshold'' result. If there is one estimation $f_i^{(j)}(X)$ exceeding $\hat{t}$, we cannot immediately declare that $f_i(X) \ge \hat{t}$. Indeed, for every fixed $f_i(X)$ and $\hat{t}$, as $\tau$ tends to infinity, with probability one, we will eventually see a noisy estimation that is above $\hat{t}$. Hence, the larger $\tau$ is, the less confident we are about the conclusion $f_i(X)\ge \hat{t}$. To ensure that we can be equally confident in the ``Above-Threshold'' case, we use the following strategy. We set a new threshold $\hat{t}_{lower} = \hat{t} - \frac{6\log(m)}{\eps'}$. Then we test if $f_i(X) \ge \hat{t}_{lower}$. Similarly, we amplify confidence by repetition. Namely, we obtain $\tau$ independent estimations of $f_i(X)$, and compare them against $\hat{t}_{lower}$. We pass the test if and only if all the $\tau$ estimations are above $\hat{t}_{lower}$. If we pass the test, then we are very confident that $f_i(X) \ge \hat{t}_{lower}$. If we fail the test, we switch back to test if $f_i(X) \le \hat{t}$ (using another $\tau$ independent estimations of $f_i(X)$). In general, we 
alternate between two tests $f_i(X) \stackrel{?}{\le} \hat{t}$ and $f_i(X) \stackrel{?}{\ge} t_{lower}$, until we pass one of them.

\paragraph{Two challenges.} To implement the idea, there are two remaining issues.
\begin{itemize}
    \item In the standard SVT algorithm, we only pay for each ``Above-Threshold'' answer, and every such answer can identify a meaningful query for us. However, in our modified version of SVT, we will repeatedly test $f_i(X) \stackrel{?}{\le} \hat{t}$ and $f_i(X) \stackrel{?}{\ge} \hat{t}_{lower}$ until we pass one of them. Each time we fail a test, we need to pay the privacy loss. How do we ensure that we can use the vast majority of our ``privacy budget'' to answer meaningful queries?
    \item In the argument above, we ignored the issue that $\hat{t}$ is only a noisy version of $t$. Indeed, if we construct $\hat{t}$ by $\hat{t} = t + \Lap(1/\eps')$, we are only $(1-\beta)$-confident that $|\hat{t} -t|\le \frac{1}{\eps'}\log(1/\beta)$. Therefore, there is a noticeable chance (e.g., with probability $10\beta$) that we will start with an erroneous noisy threshold $\hat{t}$ (e.g., $|\hat{t} - t| \ge \frac{1}{\eps'}\log(\frac{1}{10\beta})$). If this does happen, then the ``alternate testing'' algorithm does not make much sense.
\end{itemize}

We resolve the first issue by allowing a gap of $\frac{6}{\eps'}\log(m)$ between $\hat{t}_{lower}$ and $\hat{t}$. Then, for any query $f_i(X)$, at least one of the following is true:
\[
f_i(X) \le \hat{t} - \frac{3}{\eps'}\log(m), ~~~~ \text{or} ~~~~ f_i(X) \ge \hat{t}_{lower} + \frac{3}{\eps'}\log(m).
\]
We assume $\tau = \log(1/\beta) \le m$. In this case, if $f_i(X)\le \hat{t}-\frac{3}{\eps'}\log(m)$, we can pass the test $f_i(X) \stackrel{?}{\le} \hat{t}$ with probability at least $1-\frac{1}{m^2}$. If $f_i(X)\le \hat{t}-\frac{3}{\eps'}\log(m)$ and we do fail the test $f_i(X) \stackrel{?}{\le} \hat{t}$, then we say that a ``mistake'' happens and one round of SVT is ``wasted''. Similarly, we will say one round of SVT is wasted if $f_i(X) \ge \hat{t}_{lower} + \frac{3}{\eps'}\log(m)$ but we fail the test $f_i(X) \stackrel{?}{\ge} \hat{t}_{lower}$. Since the probability of making a mistake is small (i.e., the probability is at most $ \frac{1}{m^2}$), we can show that if we run our version of SVT for $k$ rounds to handle $m$ queries, then with probability $1-\beta$, at most $\frac{\log(1/\beta)}{\log(m)}$ rounds are wasted. Therefore, we can still use the algorithm to identify $k - \frac{\log(1/\beta)}{\log(m)}$ meaningful queries.

We resolve the second issue by noting that there is a version of SVT that only requires one to noisify the threshold once (see, e.g., \cite{DBLP:journals/pvldb/LyuSL17}). Here, the crucial point is that when we only add noise to the threshold once, we can add a much smaller noise to ensure privacy. However, the result from \cite{DBLP:journals/pvldb/LyuSL17} only considers pure-DP case, and the privacy loss scales linearly with $k\eps$ in their analysis. We manage to prove Theorem~\ref{theo:meta-test}, which shows that the privacy scales proportionally with $\sqrt{k}\eps$ if we relax the requirement to approximate DP. 

To be more precise, we explain the utility improvement quantitatively. In the standard SVT algorithm, we need to re-initialize the noisy threshold before starting each of the $k$ rounds, and we draw noises from the distribution $\Lap(1/\eps')$ to achieve privacy. With our new composition theorem, we only need to add one noise drawn from $\Lap(1/\eps)$. Recall that $\frac{1}{\eps'} \approx \frac{\sqrt{k\log(1/\delta)}}{\eps}$. The latter noise is considerably smaller than the former one.

In the formal proof (Section~\ref{sec:proof-SVT}), we design our algorithm under the framework of Algorithm~\ref{algo:sync}, using the $\FTest$ procedure as the main subroutine, which turns out to be more amenable to mathematical manipulations. In the discussion above, we chose to discuss our algorithmic idea under the standard SVT framework, because SVT is arguably more well-known, and it might help readers gain the intuition better.

\paragraph{The possibility of not noisifying threshold.} We also note that there are variants of SVT that do not add noise to the threshold at all (see, e.g., \cite{DBLP:conf/focs/HardtR10, DBLP:conf/colt/KaplanLMNS20}). It is tempting to apply those variants to address the second issue. However, looking into the privacy proofs for those variants, one can note that their $\delta$ parameter has to be at least $2^{-k}$, where $k$ is the number of rounds that we run SVT (i.e., the number of $\top$ responses we can get before halting the algorithm). In most applications, requiring $\delta$ to be at least $2^{-k}$ would not be an issue. However, in the application to private multiplicative weights and data analysis, we are also interested in the case where $\delta \ll 2^{-k}$. In this case, adding a noise to the threshold and using Theorem~\ref{theo:meta-test} to do the privacy accounting give us the best possible result.

More precisely, if we use a version of SVT that does not noisify the threshold, we will end up with weaker versions of Theorems~\ref{theo:MWU-new} and \ref{theo:adaptive-analysis-improved}. Specifically, the error bound in Theorem~\ref{theo:MWU-new} degrades to
\[
O\left( \frac{\sqrt{\log|\calX|\log(1/\delta)} \log(m)}{n\eps\alpha} + \frac{\log(1/\delta)\log(m)}{n\eps} \right),
\]
and the sample complexity in Theorem~\ref{theo:adaptive-analysis-improved} degrades to
\[
n \le O\left( \min\left( \frac{\sqrt{\log|\calX|\log (1/\beta)}}{\alpha^3}, \frac{\log|\calX|}{\alpha^4} \right) \log m + \frac{\log(1/\beta)}{\alpha^2}\right).
\]

\section{Proof of Privacy Theorems} \label{sec:metaprivacyproofs}

\subsection{Proof of Theorem~\ref{theo:meta-selection}}

\begin{proof}
We first consider the case that all the mechanisms $M_i$'s satisfy pure-DP (i.e., $\delta_i = 0$ for all $i$). At the end the proof, we explain how to deal with candidates with approximate DP property. Let $P_{\gamma}$ denote the CDF for the parameter $p$. Namely, $P_{\gamma}(v) = \Pr[x\le v] = v^\gamma$. Fix $c\ge 1$. Let $\calB$ be an adversary that adaptively calls $\FSelect$ for $c$ rounds. Denote Algorithm~\ref{algo:sync} as $\calA$. For each $p\in [0,1]$, let $\calA^p$ denote Algorithm~\ref{algo:sync} conditioning on it having sampled $p$ in the initialization step.

In one round, suppose that the $\calB$ calls \FSelect with $(\tau, k, M_1,\dots, M_k)$. By duplicating each mechanism for $\tau$ times, we can assume without loss of generality that $\tau = 1$. Let $(i,s)\in [k]\times \calY$ be a potential outcome. Let $D$ and $D'$ be two neighboring datasets. Let $\calA^p(D, (M_1,\dots, M_k))$ denote the output of $\calA^p$ given input $D$ and query $(M_1,\dots, M_k)$. We prove
\[
\Pr[\calA^p(D, (M_1,\dots, M_k)) = (i,s)] \le e^{2\eps} \Pr[\calA^{p/e^\eps}(D', (M_1,\dots, M_k))= (i, s)].
\]
Note that $\calA$ on this \FSelect call returns $(i,s)$, if and only if both of the following events hold:

\begin{itemize}
    \item Event $\calE_1$: $r_{i,1} = 1$ and $M_i$ outputs $s$.
    \item Event $\calE_2$: For every $i'\ne i$, either $r_{i',1} = 0$ or $M_{i'}$ outputs a solution that is inferior than $s$.
\end{itemize}
 
Now, observe that
\begin{align}
\Pr[\calE_1\mid \calA^p(D)] 
&= p\cdot \Pr[M_i(D)= s] \notag \\
&\le e^{\eps} p \Pr[M_i(D')=s] \notag \\
&\le e^{2\eps} \frac{p}{e^\eps} \Pr[M_i(D') = s] \notag \\
&= e^{2\eps} \Pr[\calE_1 \mid \calA^{p/e^\eps}(D')]. \label{eq:selection-event-1}
\end{align}
For every $i'\ne i$, we have
\[
\begin{aligned}
&~~~~ \Pr[r_{i',1}=1\land M_{i'} \text{ outputs a solution better than $s$} \mid \calA^p(D)] \\
&= p\cdot \Pr[M_{i'}(D) \text{ outputs a solution better than $s$}] \\
&\ge \frac{p}{e^{\eps}} \cdot \Pr[M_{i'}(D') \text{ outputs a solution better than $s$}] \\
&= \Pr[r_{i',1}=1\land M_{i'} \text{ outputs a solution better than $s$} \mid \calA^{p/\eps}(D')].
\end{aligned}
\]
Therefore,
\[
\begin{aligned}
& ~~~~ \Pr[r_{i',1}=0 \lor M_{i'} \text{ outputs a solution worse than $s$} \mid \calA^p(D)] \\
& \le \Pr[r_{i',1}=0 \lor M_{i'} \text{ outputs a solution worse than $s$} \mid \calA^{p/e^{\eps}}(D')].
\end{aligned}
\]
Consequently,
\begin{align}
\Pr[\calE_2 \mid \calA^p(D)] \le \Pr[\calE_2 \mid \calA^{p/e^\eps}(D')]. \label{eq:selection-event-2}
\end{align}
Combining \eqref{eq:selection-event-1} and \eqref{eq:selection-event-2} yields that
\[
\Pr[\calA^p(D, (M_1,\dots, M_k)) = (i,s)] \le e^{2\eps} \Pr[\calA^{p/e^\eps}(D', (M_1,\dots, M_k))= (i, s)].
\]
Note that this inequality holds for all $(M_1,\dots, M_k)$, \emph{independently of} the private input. Now, suppose $\calB$ interacts with $\calA^{p}(D)$ (or $\calA^{p/e^\eps}(D')$) for $c$ rounds. We can apply the argument above for the $c$ rounds separately. Let $\IT(\calB:\calA)$ denote the random variable recording the interaction transcript between $\calB$ and $\calA$. For every possible collection $E$ of outcomes, we have
\[
\Pr[\IT(\calB:\calA^{p}(D))\in E] \le e^{2c\eps} \Pr[\IT(\calB:\calA^{p/e^\eps}(D'))\in E].
\]

Finally, we have
\[
\begin{aligned}
&~~~~\Pr[\IT(\calB:\calA(D))\in E] \\
&= \int_{0}^1 \Pr[\IT(\calB:\calA^{p}(D))\in E] dP(p) \\
&\le \int_{0}^1 e^{2c\eps} \Pr[\IT(\calB:\calA^{p/e^\eps}(D'))\in E] \cdot (\gamma p^{\gamma-1}) \cdot dp \\
&\le \int_{0}^{e^\eps} e^{2c\eps+\eps} \Pr[\IT(\calB:\calA^{p/e^\eps}(D'))\in E]\cdot (\gamma (p/e^\eps)^{\gamma-1}) \cdot e^{(\gamma-1)\eps}\cdot d(p/e^{\eps}) \\
&\le \int_{0}^{e^\eps} e^{(2c+\gamma)\eps} \Pr[\IT(\calB:\calA^{p/e^\eps}(D'))\in E]\cdot dP(p/e^\eps) \\
&= e^{(2c+\gamma)\eps} \Pr[\IT(\calB:\calA(D'))\in E].
\end{aligned}
\]
This completes the proof.

Now, we consider the case that each mechanism $M_i$ is $(\eps,\delta_i)$-DP with $\delta_i \ge 0$. For two neighboring inputs $D, D'$, we can decompose $M_i(D) = (1-\delta_i) N_i(D) + \delta_i E_i(D)$ and $M_i(D') = (1-\delta_i) N_i(D') + \delta_i E_i(D')$, where $N_i(D)$ and $N_i(D')$ are $(\eps,0)$-indistinguishable and $E_i(D),E_i(D')$ are arbitrary. Then, each time we run $M_i(D)$, we can first sample from $(N_i(D), E_i(D))$ with probability $(1-\delta, \delta)$, and run the chosen mechanism. It follows that with probability at least $1-\tau \delta_i$, all the executions involve only the $N_i$-part. We union-bound over $i$, and conclude that with probability at least $1-\tau\sum_{i}\delta_i$, all the executions involve only the $N_i$-part. Conditioning on this event, the argument above shows that the outputs of the algorithm on $D,D'$ are $((2c+\gamma)\eps,0)$-DP. It follows that the original system is $((2c+\gamma)\eps,\tau\sum_{i}\delta_i)$-DP.
\end{proof}

\subsection{Proof of Theorem~\ref{theo:meta-test}}

In this section, we prove Theorem~\ref{theo:meta-test}. To ease the presentation, we will assume that all the mechanisms fed into \FTest satisfy pure-DP. Having proved for this case, the proof for approximate-DP mechanisms can be verified using similar arguments as in the proof of Theorem~\ref{theo:meta-selection}.

\subsubsection{Preliminaries}

We will use the R\'enyi Differential Privacy \cite{DBLP:conf/csfw/Mironov17-renyi} framework to prove Theorem~\ref{theo:meta-test}. Recall the definition. For two distributions $P,Q$ on the universe $X$ and a real $\alpha \in (1, \infty)$, we define the $\alpha$-order R\'enyi divergence of $P$ from $Q$ as
\[
D_{\alpha}(P \| Q) := \frac{1}{\alpha - 1} \log \left( \sum_{x\in X} P(x) \left( \frac{P(x)}{Q(x)} \right)^{\alpha - 1} \right).
\]
The max-divergence is defined as
\[
D_{\infty}(P \| Q) := \log(\max_{x} \{P(x)/Q(x)\}).
\]

Since we will defer most of the technical manipulations to Appendix \ref{appendix:proof-privacy}, we don't discuss properties of divergences in the main paper.

\subsubsection{Reduction to an Asymmetrical Coin Game}

We use $\calA$ to denote Algorithm~\ref{algo:sync}. For $p\in [0,1]$, let $\calA^p$ denote the algorithm $\calA$ conditioning on it having sampled $p$ in the initialization step. Fix two neighboring dataset $S,S'$\footnote{We use $S,S'$ to denote data sets in this subsection, since the letter ``$D$'' is reserved for divergence.}. Let $\IT(\calB:\calA)$ denote the random variable recording the interaction between the two randomized systems $\calB$ and $\calA$. Let $\calB$ be an adversary that interacts with $\calA$ by calling $\FTest$. Fix $p\in [0,1]$, we first compare $\IT(\calB:\calA^{p}(S))$ with $\IT(\calB:\calA^{p/e^{\eps}}(S'))$. It turns out that the two interactions can be captured by the following experimental game, which we call the ``$0$-favored'' coin flipping game. Our formulation of the coin-flipping game is inspired by several coin-flipping mechanisms appeared in previous works \cite{DBLP:conf/soda/GuptaLMRT10, DBLP:journals/corr/OhV13, DBLP:journals/toc/MurtaghV18, DBLP:conf/colt/KaplanMS21}.

\begin{algorithm2e}[H]
\LinesNumbered
    \caption{The $0$-favored coin-flipping mechanism.}
    \label{algo:coin-flip-asymmetric}
    \DontPrintSemicolon
    \KwIn{
           An input bit $b\in \{0,1\}$, a privacy parameter $\eps \in (0,1)$, an integer $k\ge 1$.
    }
    \SetKwProg{Init}{initialize}{:}{}
    \Init{}{
        $c\gets 0$ \;
    }
    \While{$c<k$} {
        Receive a query $(p,q)\in [0,1]^2$, promised that $0 \le q\le p\le e^{\eps} q$ and $(1-q)\le e^{\eps} (1-p)$ \;
        \If{$b=0$}{
            Sample $r\sim \Ber(p)$\;
        }
        \Else{
            Sample $r\sim \Ber(q)$\;
        }
        \If{$r=1$}{
            $c\gets c + 1$ \;
        }
        Output $r$ \;
    }
\end{algorithm2e}

We call this mechanism $0$-favored, because for each query $(p, q)$, we always have $q\le p$, meaning that the algorithm is more likely to output $1$ when its private input is $b = 0$. 

We claim that Algorithm~\ref{algo:coin-flip-asymmetric} with privacy parameter $2\eps$ can simulate $\IT(\calB:\calA^{p}(S))$ and $\IT(\calB:\calA^{p/e^\eps}(S'))$. To see this, suppose that $\calB$ knows that he is interacting with either $\calA^{p}(S)$ or $\calA^{p/e^\eps}(S')$. When $\calB$ prepares a hypothesis $H$, he can compute
\[
p_0 = \Pr[\calA^{p}(S) \text{ on } \FTest(H) \text{ outputs } \top]
\]
and
\[
p_1 = \Pr[\calA^{p/e^\eps}(S') \text{ on } \FTest(H) \text{ outputs } \top].
\]
It follows that $p_0 = p\cdot \Pr[H(S) = \top] \ge \frac{p}{e^\eps} \Pr[H(S') = \top] = p_1$, $p_0 \in (e^{-2\eps}p_1, e^{2\eps}p_1)$ and $(1-p_0)\in (e^{-2\eps}(1-p_1),e^{2\eps}(1-p_1))$. Therefore, sending a query $H$ to Algorithm~\ref{algo:sync} is equivalent to sending a query $(p_0,p_1)$ to Algorithm~\ref{algo:coin-flip-asymmetric}.

\subsubsection{Analysis of the Coin Game}

Algorithm \ref{algo:coin-flip-asymmetric} appears not to be differentially private. Nevertheless, it satisfies a ``one-sided'' stability property, which can be captured by max-divergence and R\'enyi divergence, as shown in the following two lemmas.

\begin{lemma}\label{lemma:max-divergence-one-sided}
Consider Algorithm~\ref{algo:coin-flip-asymmetric}. Fix $\calB$ to be an arbitrary adversary interacting with Algorithm~\ref{algo:coin-flip-asymmetric}. Let $P,Q$ denote the distributions of the interaction between $\calA$ and Algorithm~\ref{algo:coin-flip-asymmetric} when the private input is $0$ or $1$, respectively. Then, for any $\alpha \in (1, \infty)$, we have $D_{\infty}(P\|Q) \le k\eps$.
\end{lemma}

\begin{lemma}\label{lemma:divergence-one-sided}
Consider the same setup as in Lemma~\ref{lemma:max-divergence-one-sided}. For any $\alpha \in (1, \infty)$, we have $D_{\alpha}(P\|Q) \le 3k \alpha  \eps^2$. 
\end{lemma}

Intuitively, the two lemmas hold because for each query $(p,q)$ with $q\le p$, the algorithm is \emph{always} more likely to output $0$ when the private input is $b=1$. Hence, when we consider $D_{\infty}(P\|Q)$ (i.e., when we consider the divergence of $P$ from $Q$), we can ``pass'' the ``$0$'' outputs ``for free''. Essentially, the divergence increases only when we observe ``$1$'' outputs. The counter $c$ counts the number of ``$1$'' outputs that we have seen so far. Since we halt the algorithm once $c$ reaches a pre-defined threshold $k$, naturally, one would expect that the divergence bound of $P$ from $Q$ grows only with $k$ (and is independent of the number of queries). Lemmas~\ref{lemma:max-divergence-one-sided} and \ref{lemma:divergence-one-sided} confirm this intuition.

The formal proof for Lemma~\ref{lemma:divergence-one-sided}, however, is quite delicate and involved. We defer proofs of both lemmas to Appendix~\ref{appendix:proof-privacy}.

\subsubsection{Wrap-up}

Now, let $\calB$ be an adversary interacting with $\calA^p(S)$ or $\calA^{p/e^\eps}(S')$. Lemma~\ref{lemma:max-divergence-one-sided} shows that for every collection of outputs $E$, it holds that
\[
\Pr[\IT(\calB:\calA^p(S))\in E] \le e^{2c\eps} \Pr[\IT(\calB:\calA^{p/e^\eps}(S'))\in E].
\]
By H\"older's inequality and Lemma~\ref{lemma:divergence-one-sided} (see also the conversion from R\'enyi DP to approximate DP, e.g., \cite{DBLP:conf/csfw/Mironov17-renyi}), there is an absolute constant $D \ge 1$ such that
\[
\Pr[\IT(\calB:\calA^p(S))\in E] \le e^{2D\sqrt{c\log(1/\delta)}\eps} \Pr[\IT(\calB:\calA^{p/e^\eps}(S'))\in E] + \delta.
\]
To finish the proof of Theorem~\ref{theo:meta-test}, we integrate over $p\in [0,1]$. This step is the same as the proof of Theorem~\ref{theo:meta-selection} and we do not repeat it here.

\subsubsection{Using \textsc{Selection} and \textsc{Test} Concurrently}

We claimed in the introduction that our framework allows an analyst to interleave its calls to $\FSelect$ and $\FTest$, and only pay the privacy loss for every call to $\FSelect$ and every $\top$ response received from $\FTest$. Specifically, if there are $c_1$ calls to $\FSelect$ and $c_2$ calls with $\top$ output to  $\FTest$, each with mechanisms that are $(\eps,0)$-DP,  then the total privacy cost is $((2c_1 + 2c_2 +\gamma)\eps,0)$-DP.

To see why this is true, let $\calB$ be an adversary that interacts with $\calA$ arbitrarily. Consider the interaction between $\calB$ and $\calA^{p}(S), \calA^{p/e^\eps}(S')$. Combining the proofs of Theorems~\ref{theo:meta-selection} and \ref{theo:meta-test}, we have
\[
\Pr[\IT(\calB:\calA^{p}(S))\in E] \le e^{2c_1\eps + 2c_2\eps} \Pr[\IT(\calB:\calA^{p/e^\eps}(S')\in E].
\]
Integrating over $p\in [0,1]$ shows that
\[
\Pr[\IT(\calB:\calA(S))\in E] \le e^{2c_1\eps + 2c_2\eps+\gamma\eps} \Pr[\IT(\calB:\calA(S')\in E],
\]
as desired.

\section{Proofs for Applications}

In this section, we prove theorems listed in Sections~\ref{sec:appLK} and \ref{sec:appselect}.

\paragraph*{Notation.} For $\eps,\delta> 0$, we define the truncated Laplace mechanism $\TLap(\eps,\delta)$ as follows. The support of $\TLap(\eps,\delta)$ is $[-\frac{\log(1/\delta)}{\eps},\frac{\log(1/\delta)}{\eps}]$. For every $v\in [-\frac{\log(1/\delta)}{\eps},\frac{\log(1/\delta)}{\eps}]$, we have $\Pr[\TLap(\eps,\delta) = v]\propto e^{-|v|\eps}$. The truncated Laplace mechanism can answer a sensitivity-$1$ query $f:\calX^n\to \mathbb{R}$ by publishing $f(S) + \TLap(\eps,\delta)$. It is well known that this mechanism is $(\eps,\delta)$-DP \cite{DBLP:conf/aistats/GengDGK20-truncated}.

\subsection{Better-Than-Median Selection}\label{sec:proof-LT}

In this section, we prove Theorem~\ref{theo:privacy-runtime-tradeoff}, restated below. 

\begin{reminder}{Theorem~\ref{theo:privacy-runtime-tradeoff}}
For every $\alpha\in (0,\infty)$ and $\beta \in (0, 1)$, for $(\eps,\delta)$-DP $\mathcal{A}$ with $\eps\in (0,1)$ and $\delta\ge 0$, there is a better-than-median algorithm $\calA^*$ with confidence $1-\beta$ that satisfies the following: 
\begin{itemize}
\item With probability $1$, $\calA^*$ makes at most $T$ oracle calls to $\calA$, where
\[
T = \begin{cases}
~~\left\lceil \frac{2}{\beta} \right\rceil & \text{if $\alpha = 1$} \\
\left\lceil 5(\frac{2}{\beta})^{1/\alpha} \log(1/\beta)\right\rceil & \text{otherwise}
\end{cases}.
\]
\item $\calA^*$ is $((2+\alpha)\eps, T\delta)$-differentially private.
\end{itemize}
Furthermore, if $\calA$ is $(\eps,\delta)$-DP with $\delta > 0$, then $\calA^*$ is $(\eps,T\delta)$-DP.
\end{reminder}

\begin{proof}
We design $\calA^*$ that initializes 
Algorithm~\ref{algo:sync} with $\gamma = \alpha$ and
calls \FSelect with arguments $(k = 1, \tau = T, M_1 = \calA)$ where $T$ given by the theorem statement. If follows from Theorem~\ref{theo:meta-selection} that $\calA^*$ is $((2+\alpha)\eps,\tau \delta)$-DP. As for the utility, we have shown the utility guarantee for $\alpha = 1$. When $\alpha \ne 1$, we argue as follows. 

First, with probability at least $1-(\beta/2)$, we have that $p\ge (2\beta)^{1/\alpha}$. We condition on this event. By our choice of $T$, we know that $\calA$ will be run for at least $t\ge \log(4/\beta)$ times with probability at least $1-(\beta/4)$. We further condition on this event. Then, the probability that the $t$ calls of $\calA$ fail to yield a better-than-median solution is at most $1-(\beta/4)$. Overall, the probability that $\calA*$ fails to output a better-than-median solution is at most $1-\beta$, as desired.
\end{proof}

\subsection{Query Releasing}\label{sec:proof-query-release}

In this section, we prove Theorem~\ref{theo:release-k-queries-non-adaptive}, restated below.

\begin{reminder}{Theorem~\ref{theo:release-k-queries-non-adaptive}}
There is a constant $C > 0$ such that the following is true. For every $\eps \in (0,1),\delta \in (0,1/2)$ and $k\in \mathbb{N}$, there is an $(\eps, \delta)$-DP algorithm that answers $k$ given sensitivity-$1$ queries $f_1,\dots, f_k$, such that
\[
\Pr_{\tilde{f}}\left[\| f - \tilde{f} \|_{\infty} \le C \frac{\sqrt{k\log 1/\delta}  + \log(1/\delta)}{\eps} \right] = 1,
\]
where $\tilde{f} := (\tilde{f}_1,\dots, \tilde{f}_k)$ denotes the responses released by the algorithm.
\end{reminder}

To prove Theorem~\ref{theo:release-k-queries-non-adaptive}, we need the following theorem due to \cite{DBLP:conf/colt/Ghazi0M21}.

\begin{theorem}[\cite{DBLP:conf/colt/Ghazi0M21}]\label{theo:GKM-algorithm}
There is a constant $C_1 > 0$ such that the following is true. For every $\eps \in (0,1),\delta \in (0,1/2)$ and $k\in \mathbb{N}$, there is an $(\eps, \delta)$-DP algorithm that, given $k$ non-adaptive sensitivity-$1$ queries $f_1,\dots, f_k$, returns a list of answers $(\tilde{f}_1,\dots, \tilde{f}_k)$, such that
\[
\Pr_{\tilde{f}}\left[\| f - \tilde{f} \|_{\infty} \le C_1 \frac{\sqrt{k\log 1/\delta} }{\eps} \right] \ge \frac{1}{2}.
\]
\end{theorem}

The following statement follows as a simple corollary of Theorem~\ref{theo:GKM-algorithm}.

\begin{corollary}\label{coro:query-release-error-lb}
There is a constant $C_2 > 0$ such that the following is true. For every $\eps \in (0,1),\delta \in (0,1/2)$ and $k\in \mathbb{N}$, there is an $(\eps, \delta)$-DP algorithm that, given $k$ non-adaptive sensitivity-$1$ queries $f_1,\dots, f_k$, returns a list of answers $(\tilde{f}_1,\dots, \tilde{f}_k)$ and a real $s\in \mathbb{R}^{+}$ satisfying the following.
\begin{itemize}
    \item With probability $1$, it holds that $\| f - \tilde{f} \|_{\infty} \le s$.
    \item With probability at least $\frac{1}{2}$, it holds that
    \[
    s\le C_2 \left( \frac{\sqrt{k\log 1/\delta} }{\eps} + \frac{\log(1/\delta)}{\eps} \right).
    \]
\end{itemize}
\end{corollary}

\begin{proof}
Let $\calA_1$ be an instantiation of the algorithm from Theorem~\ref{theo:GKM-algorithm} with privacy parameter set to $(\eps/2,\delta/2)$. Given a list of queries $(f_1,\dots, f_k)$. We first run $\calA_1$ to produce a list of responses $(\tilde{f}_1,\dots, \tilde{f}_k)$. Let $g = \|f-\tilde{f}\|_\infty$. We compute $s = g + 2\frac{\log(1/\delta)}{\varepsilon} + \TLap(\varepsilon/2,\delta/2)$. Finally, we release $(\tilde{f}_1,\dots, \tilde{f}_k)$ and $s$. It follows that $s$ is an upper bound for $\|f-\tilde{f}\|_2$ with probability $1$. Moreover, since $\|f-\tilde{f}\|_\infty \le 4C_1\frac{\sqrt{k\log 1/\delta} }{\eps}$ happens with probability at least $\frac{1}{2}$, it follows that $s\le \max(4C_1,4)\cdot \left(\frac{\sqrt{k\log 1/\delta} }{\eps} + \frac{\log(1/\delta)}{\varepsilon} \right)$ with probability at least $\frac{1}{2}$. Since the truncated Laplace mechanism is $(\eps/2,\delta/2)$-DP, by the basic composition theorem, the whole algorithm is $(\eps,\delta)$-DP.
\end{proof}

Given Corollary~\ref{coro:query-release-error-lb}, Theorem~\ref{theo:release-k-queries-non-adaptive} follows by tuning parameters properly.

\begin{proofof}{Theorem~\ref{theo:release-k-queries-non-adaptive}}
We will combine Corollary~\ref{coro:query-release-error-lb} with Theorem~\ref{theo:privacy-runtime-tradeoff}.
In particular, given the target privacy budget $(\eps,\delta)$, let $\calA$ be an instantiation of Corollary~\ref{coro:query-release-error-lb} with privacy parameters $(\eps/3,\delta^2/10)$. Then we use Theorem~\ref{theo:privacy-runtime-tradeoff} on top of $\calA$ with $\gamma = 1$ and $\beta = \delta/10$, to select a vector $(\tilde{f}_1,\dots, \tilde{f}_k)$ with the smallest associated quality score $s$. Let $C_2 = \max(4C_1,4)$. If $s \le C_2\left(\frac{\sqrt{k\log 1/\delta} }{\eps} + \frac{\log(1/\delta)}{\varepsilon} \right)$, then we simply output $(\tilde{f}_1,\dots, \tilde{f}_k)$. Otherwise, we non-privately output the original vector $(f_1,\dots, f_k)$. Since $s$ is an upper bound of $\|f-\tilde{f}\|_{\infty}$, the utility guarantee is evident.

It remains to verify that the algorithm is $(\eps,\delta)$-DP. The vector $(\tilde{f}_i)_{i\in [k]}$ selected by Theorem~\ref{theo:privacy-runtime-tradeoff} is $(\eps,\delta/5)$-DP with respect to the private input. Moreover, with probability at least $1-\delta/5$, it holds that
\[
s \le C_2\left(\frac{\sqrt{k\log 1/\delta} }{\eps} + \frac{\log(1/\delta)}{\varepsilon} \right).
\]
Therefore, we will do the non-private correction with probability at most $\delta/5$. Consequently, the whole algorithm $(\eps,\delta)$-DP. 
\end{proofof}

\subsection{Top-$k$ Selection}\label{sec:proof-top-k}

In this section, we prove Theorem~\ref{theo:top-k-selection}, restated below.

\begin{reminder}{Theorem~\ref{theo:top-k-selection}}
There is an absolute constant $C>0$ such that the following holds. For every $\eps,\delta,\beta\in (0, 1)$, there is an $(\eps,\delta)$-DP algorithm for top-$k$ selection satisfying the following. Let $\mathbf{S}$ denote the output of the algorithm. Then
\[
\Pr\left[ \Gap(\mathbf{S})\ge C\frac{\sqrt{k\log(1/\delta)}}{\eps} \cdot \log(m) + C\frac{\log(1/\beta)}{\varepsilon} \right] \le \beta.
\]
\end{reminder}

\begin{proof}
We first consider the case that $\beta \ge \delta$. Define a mechanism $\calA$ as follows. Set $\eps' = \frac{\eps}{40 \sqrt{k\log(1/\delta)}}$, $\calA$ repeatedly run the exponential mechanism for $k$ rounds, to select $k$ candidates $\mathbf{S}\subset [m]$. Then $\calA$ calculates $\tilde{q}(S) = \Gap(S) + \Lap(\eps/6) + 13\frac{\log(1/\beta)}{\eps}$ and outputs $\mathbf{S}$ and $\tilde{q}(S)$. By the basic and advanced composition theorems of DP \cite{DBLP:conf/focs/DworkRV10}, $\calA$ is $(\eps/3,\delta^2/10)$-DP.

With probability $1-\beta/2$, all the Laplace noises sampled in the execution of $\calA$ are bounded by $\frac{10\log(1/\beta)}{\eps}$. We condition on this event. Then, it is easy to see that $\tilde{q}(S)$ provides an upper bound of $\Gap(\mathbf{S})$. Moreover, by the utility proof of the exponential mechanism, there is an absolute constant $C > 0$ such that
\[
\P_{\mathbf{S},\tilde{q}}\left[ \tilde{q}(\mathbf{S}) \le C\frac{\sqrt{k\log(1/\delta)}}{\eps} \cdot \log(m) + C\frac{\log(1/\beta)}{\varepsilon}  \right] \ge \frac{1}{2}.
\]
We use Theorem~\ref{theo:privacy-runtime-tradeoff} on top of $\calA$ with $\gamma = 1$ and $\beta = \delta/10$, to select a subset $\mathbf{S}$ with the smallest associated $\tilde{q}(S)$. The utility property of Theorem \ref{theo:privacy-runtime-tradeoff} shows that with probability $1-\beta$, we have 
\[
\Gap(\mathbf{S}) \le \tilde{q}(S)\le C\frac{\sqrt{k\log(1/\delta)}}{\eps} \cdot \log(m) + C\frac{\log(1/\beta)}{\varepsilon}.
\]
The privacy property of Theorem~\ref{theo:privacy-runtime-tradeoff} shows that the result algorithm is $(\eps,\delta)$-DP, as desired.

For the case $\beta < \delta$, we can first design an $(\eps,\delta')$-DP top-$k$ selection algorithm with confidence $\beta'$, where $\delta' = \delta/10$ and $\beta' = \delta'$. It follows that with probability at least $1-\beta'$, we have
\[
\Gap(\mathbf{S})\le C'\frac{\sqrt{k\log(1/\delta)}}{\eps} \cdot \log(m) + C'\frac{\log(1/\beta')}{\varepsilon} 
\]
for a constant $C' > 0$. We can use the non-private correction trick similar to the proof of Theorem~\ref{theo:release-k-queries-non-adaptive} to non-privately correct bad subset $\mathbf{S}$. This degrades the privacy bound of the algorithm from $(\eps,\delta/10)$-DP to $(\eps,\delta)$-DP, and gives the desired confidence bound (in fact, the confidence is improved to with probability one).
\end{proof}

\subsection{Stable Selection}\label{sec:proof-stable-selection}

In this section, we prove the results for stable selection.

\paragraph*{Choosing Mechanism.} Let $\mathcal{F}$ be a family of $1$-Lipschitz functions. Recall that $\calF$ is $k$-bounded, if for any two neighboring datasets $D,D'$, it holds that $\sum_{f\in \calF}|f(D)-f(D')|\le k$. We now prove Theorem~\ref{theo:choosing-mechanism}, which is restated below.

\begin{reminder}{Theorem~\ref{theo:choosing-mechanism}}
Suppose $\calF$ is a $k$-bounded function family. Then, there is an $(\eps,\delta)$-DP algorithm $\calA$ that receives a private input $X$ and selects a function $f\in \calF$ such that, with probability at least $1-\beta$,
\[
f(D) \ge \max_{f^*\in \calF}\{ f^*(D)\} - O(\frac{1}{\eps} \log(k/\delta\beta)).
\]
\end{reminder}

\begin{proof}
For each function $f_i\in \calF$, construct a mechanism $\calM_i$ as follows: $\calM_i$ receives the takes as input the dataset $D$ and outputs a pair $(i, f_i(D) + \TLap(\eps, \frac{\delta\beta}{5k}))$.

We use the \textsc{Selection} function in Algorithm~\ref{algo:sync} with parameters $\gamma = 1$, $\tau = \frac{4}{\beta}$ and mechanisms $\{ M_i \}_{f_i\in \calF}$, and output the index $i$ returned from the \textsc{Selection}.

To see the privacy, we fix an arbitrary pair of adjacent datasets $D,D'$. It follows that $\sum_{f\in \calF}|f(D) - f(D')| \le k$. Moreover, it is easy to see that that $i$-th mechanism is $(\eps,\frac{\delta\beta}{5k} \cdot |f_i(D)-f_i(D')|)$-DP. Therefore, by Theorem~\ref{theo:meta-selection}, our mechanism is $(3\eps, \delta)$-DP. 

To see the utility, fix the dataset $D$ and let $i^*$ be the index of the best candidate. Then, with probability at least $1-\beta$, at least one trial of $M_{i^*}(D)$ returns a pair $(i,s)$ with $s\ge f_{i^*}(D)$. We condition on this event. Let $(j,s')$ be the result returned from \textsc{Selection}. It follows that $f_j(D)\ge s' - \frac{5}{\eps}\log(k/\beta\delta) \ge s - \frac{5}{\eps}\log(k\beta/\delta) \ge f_{i^*}(D) - \frac{5}{\eps}\log(k/\beta\delta)$, as desired.
\end{proof}

\paragraph*{Stable Selection.} Suppose $\calF$ is a function family. For any given dataset $D$, sort functions in $\calF$ by decreasing order of $f(D)$. Namely, write $\calF = \{f_1,f_2,\dots, f_{m}\}$ so that $f_1(D) \ge f_2(D) \ge \dots \ge f_m(D)$. Recall that we have defined the $k$-th gap of $D$ on $\calF$ as $G_k(D,\calF) = f_1(D) - f_{k+1}(D)$. We prove Theorem~\ref{theo:stable-selection}, restated below.

\begin{reminder}{Theorem~\ref{theo:stable-selection}}
For any $k\ge 1$ and $\beta \in (0, 1)$, there is an $(\eps,\delta)$-DP algorithm $\calA$ that holds a dataset $D$ and achieves the following. For any $1$-Lipschitz function family $\calF$, suppose $G_k(D,\calF)\ge \frac{10}{\eps}\log(k/\delta\beta)$. Then, with probability at least $1-\beta$, $\calA$ selects a function $f\in \calF$ such that
\[
f(D) \ge \max_{f^*\in \calF}\{ f^*(D)\} - O(\frac{1}{\eps}\log(k/\delta\beta)).
\]
\end{reminder}

\begin{proof}
Let $\eps' = \eps/3$. Given $\calF = \{f_1,\dots, f_{m}\}$, $k\ge 1$ and the dataset $D$, let $Q(D)$ be the $(k+1)$-th largest score $f_i(D)$ among $\{f_i(D)\}_{i\in [m]}$. We define a list of mechanisms $M_1,\dots, M_m$, where $M_i(D)$ just returns index $i$ with score $s_i = \max(f_i(D) - Q(D), 0) + \TLap(\eps',\frac{\beta\delta}{10k})$. We initialize Algorithm~\ref{algo:sync} with $\gamma = 1$, and call \textsc{Selection} with $(k = m, \tau = \lceil \frac{2}{\beta} \rceil, M_1,\dots, M_m)$ to select an index $i^*$ with the largest score. Then we output $i^*$.

To see the utility, let $j^* = \arg\max_{j\in [m]} f_j(D)$. If $G_k(D,\calF) \ge 10\log(k/\delta \beta)$. It follows that $f_j(D) - Q(D) \ge 10\log(k/\delta \beta)$. By our choice of $\tau$, with probability at least $1-\beta$, $M_{J^*}$ produces at least one pair $(j^* s_{j^*})$ with $s_{j^*} \ge f_{j^*}(D) - Q(D)$. We condition on this event. Let $i^*$ be the index of returned by the algorithm. Since $(j^*,s_{j^*})$ is in the candidate set, it follows that $\max(f_{i^*}(D)-Q(s),0) \ge f_{j^*}(D)-Q(D) - \frac{1}{\eps'}\log(k/\delta\beta) > 0$. Therefore, $f_{i^*}(D)\ge f_{j^*}(D) - \frac{2}{\eps'}\log(k/\delta\beta)$, as desired.

Now we prove the privacy bound. Fix $D,D'$ to be two neighboring inputs. We first observe that $|Q(D) - Q(D')| \le 1$. For every $i\in [m]$, consider $M_i(D)$ and $M_i(D')$. There three cases:
\begin{itemize}
    \item[Case $1$.] Both $f_i(D) > Q(D)$ and $f_{i}(D') > Q(D')$. Then it follows that $|(f_i(D) -Q(D)) - (f_i(D') - Q(D'))|\le 2$. Therefore, $M_i(D)$ and $M_i(D')$ are $(\eps', \frac{\delta\beta}{5k})$-indistinguishable (in the DP sense).
    \item[Case $2$.] $f_i(D) > Q(D)$ but $f_{i}(D') \le Q(D')$. Still, it is true that $|(f_i(D) -Q(D)) - 0| \le |(f_i(D) -Q(D)) - (f_i(D') - Q(D'))|\le 2$. Therefore, $M_i(D)$ and $M_i(D')$ are $(\eps', \frac{\delta\beta}{5k})$-indistinguishable. The case that $f_i(D)\le Q(D)$ but $f_i(D') > Q(D')$ is similar.
    \item[Case $3$.] Both $f_i(D) \le Q(D)$ and $f_{i}(D') \le Q(D')$. Then it follows that $M_i(D)$ and $M_i(D')$ are $(\eps',0)$-indistinguishable.
\end{itemize}
By our definition of $Q(D)$, there are at most $2k$ indices $i\in [m]$ that belong to Cases 1 and 2. Therefore, by Theorem~\ref{theo:meta-selection}, the whole algorithm is $(3\eps',\tau \frac{\delta\beta \cdot 2k}{5k})$-DP. Since $\tau = 2/\beta$ and $\eps' = \eps/3$, it follows that the algorithm is $(\eps, \delta)$-DP.
\end{proof}

\section{Improved Sparse Vector Technique}\label{sec:proof-SVT}

In this section, we present our improved sparse vector algorithm, and apply it to improve the private multiplicative weight update algorithm.

\subsection{The Repetitive SVT}

Before showing the algorithm, we set up some pieces of notation. In SVT, we want to answer queries of the form $f_i(D) \stackrel{?}{\le} t$, where $f_i(D)$ is a function of sensitivity $\Delta$. 

We formulate a private algorithm $U_\eps(f_i,t)$ to test the hypothesis $f_i(D) \lessapprox t$ privately. $U_\eps(f_i,t)$ computes $\hat{f}_i(D) = f_i(D) + \Lap(\Delta/\eps)$, returns $\top$ if $\hat{f}_i(D) \ge t$, and returns $\perp$ otherwise. It is easy to see that $U_\eps(f_i,t)$ is $(\eps,0)$-DP for every sensitivity-$1$ function $f_i$. For our purpose, we also need a private algorithm to test lower bounds of $f_i(D)$. Let $d > 0$ be a parameter. We design a private algorithm $V_{d,\eps}(f_i,t)$ to test $f_i(D) \gtrapprox t - d$. $V_{d,\eps}(f_i,t)$ computes $\hat{f}_i(D) = f_i(D) + \Lap(\Delta/\eps)$, returns $\top$ if $\hat{f}_i(D) \le t-d$, and returns $\perp$ otherwise.

\paragraph*{The algorithm.} The following shows our improved SVT algorithm, which we call the Repetitive SVT. As we have discussed in Section~\ref{sec:apptest}, it works by alternating the testing of two hypotheses $U_{\eps}(f_i,t)$ and $V_{d,\eps}(f_i,t)$. Here, $U_{\eps}(f_i,t)$ tests the truthfulness of $f_i(D) \lessapprox t$. If $U_{\eps}(f_i,t)$ outputs $\top$, then we say the test is ``failed'', and the hypothesis $f_i(D)\lessapprox t$ is ``refuted'', otherwise we say the test passes. Similarly, $V_{d,\eps}(f_i,t)$ tests the validity of $f_i(D) \gtrapprox t-d$. 

\begin{algorithm2e}[h]
\LinesNumbered
    \caption{The Repetitive SVT}
    \label{algo:repetitive-SVT}
    \DontPrintSemicolon
    \KwIn{
        Parameters $\eps\in (0, 1), \gamma > 0, d > 0, k,\tau \in \mathbb{N}$.
        Private input $D\in \calX^n$. 
    }
    \SetKwProg{Init}{initialize}{:}{}
    \Init{}{
        Initialize Algorithm~\ref{algo:sync} with $\gamma$. \;
    }
    \While{$c<k$} {
        Receive a query $(f_i,t)$ \tcp*[f]{We want to test whether $f_i(D) \lessapprox t$}\;
        $D\gets U_{\eps}$ \tcp*[f]{$D$ will alternate between $U_{\eps}$ and $V_{d,\eps}$}\;
        \While{True}{
            $\mathsf{success} \gets 1$ \;
            \For(\tcp*[f]{Test $D$ independently for $\tau$ times}){$i=1,\dots, \tau$}{
                $\Gamma \gets \FTest(D(f_i,t))$ \tcp*[f]{``\FTest'' is the procedure in Algorithm~\ref{algo:sync}} \;
                \If{$\Gamma = \top$}{
                    $\mathsf{success} \gets 0$ \;
                    break \;
                }
            }
            \If(\tcp*[f]{Passed all of $\tau$ tests}){$\mathsf{success} = 1$}{
                break  \;
            }
            \Else(\tcp*[f]{Failed one test; switch to test the other direction}){
                $c\gets c + 1$  \;
                \If{$c = k$}{
                    HALT the algorithm \;
                }
                \If{$D = U_{\eps}$}{
                    $D\gets V_{d,\eps}$ \;
                }
                \Else{
                    $D\gets U_{\eps}$ \;
                }
            }
        }
        \If{$D = U_{\eps}$}{
            Output $\perp$ \tcp*[f]{Report $f_i(D) \lessapprox t$} \;
        }
        \Else{
            Output $\top$ \tcp*[f]{Report $f_i(D) \gtrapprox t - d$} \;
        }
    }
\end{algorithm2e}

\subsection{Analysis}

In this section, we analyze the privacy and utility properties of Algorithm~\ref{algo:repetitive-SVT}. The privacy follows as a direct corollary of Theorem~\ref{theo:meta-test}.

\begin{lemma}\label{lemma:SVT-privacy}
Algorithm~\ref{algo:repetitive-SVT} is $((2k+\gamma)\eps,0)$-DP and $(O(\gamma\eps + \sqrt{k\log(1/\delta)}\eps),\delta)$-DP for all  $\delta > 0$.
\end{lemma}

Next, we show how to tune parameters $\gamma,d,\tau$ properly to get the best possible utility. Let the target confidence bound by $1-\beta$. We assume $\beta \in (2^{-m},1/m)$. The upper bound of $\beta$ is natural, as our improvement is significant only when $\beta \le m^{-\omega(1)}$. The lower bound on $\beta$ is a rather mild restriction.

Suppose our final privacy budget is $(O(\eps),O(\delta))$-DP, and we need to use Algorithm~\ref{algo:repetitive-SVT} to process at most $m$ queries. Among them, it is promised that there are at most $k$ ``Above-Threshold'' queries. Then, we set
\[
\begin{cases}
\gamma = \frac{\log(20/\beta)}{\log(m)},\\
\tau = 5m^2, \\
\eps' = \frac{\eps}{\gamma + \sqrt{k\log(1/\delta)}}, \\
k' = k + 7\frac{\log(1/\beta)}{\log(m)}, \\
d = \frac{10\Delta}{\eps'} \log(m).
\end{cases}
\]
We run Algorithm~\ref{algo:repetitive-SVT} with parameters $\eps',\gamma, d,\tau, k'$ to process the queries. Lemma~\ref{lemma:SVT-privacy} implies that Algorithm~\ref{algo:repetitive-SVT} with the parameters as specified above enjoys $(\xi,\delta)$-DP, where
\[
\xi = O\left(\eps' \sqrt{k'\log(1/\delta)} \right) \le O(\eps).
\]

We prove the following lemma concerning the utility guarantee of Algorithm~\ref{algo:repetitive-SVT}.

\begin{lemma}\label{lemma:SVT-utility}
For the parameters as stated above. With probability $1-\beta$, Algorithm~\ref{algo:repetitive-SVT} achieves the following utility guarantee.
\begin{itemize}
    \item For every query $(f_i,t)$ that received a $\perp$ response, it holds that $f_i(D) \le t$.
    \item For every query $(f_i,t)$ that received a $\top$ response, it holds that $f_i(D) \ge t - d$.
    \item Either Algorithm~\ref{algo:repetitive-SVT} processes all $m$ queries and halts, or it halts earlier. In the latter case, Algorithm~\ref{algo:repetitive-SVT} must has outputted at least $k$ $\top$ responses before halting.
\end{itemize}
\end{lemma}

\begin{proof}
Let $p$ denote the ``pass'' probability sampled during the initialization. It follows that with probability at least $1-\beta/20$, we have $p\ge \frac{1}{m}$. We condition on this event.

We begin with verifying the first two bullets. Note that if a $\perp$ response is returned on a query $(f_i, t)$, it implies that the algorithm passed $\FTest(U_{\eps'}(f_i,t))$ for $\tau=5m^2$ times. Suppose $f_i(D) > t$. Then, the probability that one execution of $\FTest(U_{\eps'}(f_i,t))$ passes is at most $1 - p\cdot \frac{1}{2} \le 1-\frac{1}{2m}$. Therefore, the probability that all of the $\tau$ independent calls to $\FTest(U_{\eps'}(f_i,t))$ pass is at most $(1-\frac{1}{2m})^{\tau} \le 2^{-2m}$. If this extremely-unlikely event does happen, then we say that a fatal error happens. Since the batch test (The ``for'' loop starting at Line 9) is run for at most $m+k'$ times. The probability that there exists a fatal error is at most $2^{-2m}\cdot (m+k') \le \beta/20$. Similarly, we can show that up to failure probability $\beta/20$, the second bullet holds.

It remains to verify the last bullet. First, note that given a query $(f_i,t)$, we know that either $f_i(D) \le t - \frac{d}{2}$ or $f_i(D) > t - \frac{d}{2}$. During the execution of Algorithm~\ref{algo:repetitive-SVT}, we say that a ``mistake'' happens, if one of the following events happens.
\begin{itemize}
    \item $f_i(D)\le t-\frac{d}{2}$, but Algorithm~\ref{algo:repetitive-SVT} fails a batch of $\FTest(U_{\eps'}(f_i,t))$ tests.
    \item $f_i(D)\ge t-\frac{d}{2}$, but Algorithm~\ref{algo:repetitive-SVT} fails a batch of $\FTest(V_{d,\eps'}(f_i,t))$ tests.
\end{itemize}
We now show that a mistake happens with low probability. We consider the first type of mistake. The second type can be argued similarly.

Recall that $d = \frac{10\Delta}{\eps'}\log(m)$. Given $f_i(D)\ge t- \frac{d}{2}$, we know that one call to $\FTest(U_{\eps'}(f_i,t))$ passes with probability at least $1-\frac{1}{m^5}$. Hence, the probability that all the $\tau$ calls pass is at least $1-\frac{5}{m^3}$. Therefore, the probability that a mistake happens is at most $\frac{5}{m^3}$.

Let $R = \lceil \frac{3\log(1/\beta)}{\log(m)}\rceil$. Since Algorithm~\ref{algo:repetitive-SVT} runs the batch test for at at most $(k+m)$ times, the probability that the number of mistakes is larger than $R$ is at most
\[
(m+k)^R \cdot \left( \frac{5}{m^3} \right)^{R} \le \left( \frac{10}{m^2} \right)^R \le \beta/10.
\]
This bound can be verified by union-bounding over all $R$-subsets of $[m+k]$.

We condition on the event that Algorithm~\ref{algo:repetitive-SVT} makes no more than $R$ mistakes. Now, suppose that Algorithm~\ref{algo:repetitive-SVT} halts before processing all the queries. For each $i\in [m]$, let $c_i\ge 1$ be the number of batch tests that Algorithm~\ref{algo:repetitive-SVT} runs for the $i$-th query. If the algorithm halts before the $i$-th query, set $c_i = 1$. Then, it follows that
\[
\#\{\text{mistakes}\} = \sum_{i=1}^m \lfloor \frac{c_i - 1}{2} \rfloor \le R.
\]
Since Algorithm~\ref{algo:repetitive-SVT} halts earlier, we have
\[
k' = \sum_{i=1}^m (c_i -1).
\]
We observe that
\[
\begin{aligned}
\#\{\text{$\top$ responses}\} 
&= \sum_{i=1}^m \mathbbm{1}\{c_i \bmod 2 = 0\} \\
&= \sum_{i=1}^m (c_i - 1) - 2\lfloor \frac{c_i - 1}{2} \rfloor \\
&= \left( \sum_{i=1}^m (c_i -1) \right) - 2 \left( \sum_{i=1}^m \lfloor \frac{c_i - 1}{2} \rfloor \right) \\
&\ge k' - 2R \\
&\ge k,
\end{aligned}
\]
as desired.
\end{proof}

To conclude, we review the confidence-accuracy tradeoff implied by Lemma~\ref{lemma:SVT-utility}. Fix the privacy budget to be $(\eps,\delta)$. Suppose that one needs to identify $k$ significant queries from a total of $m$ queries, and one desires $(1-\beta)$ confidence bound. Then, the additive error bound that Algorithm~\ref{algo:repetitive-SVT} can offer is
\[
d\le O\left( \Delta \frac{\sqrt{k\log(1/\delta)}\log(m) + \log(1/\beta)}{\eps} \right).
\]

\begin{remark}\label{remark:SVT-pure-DP}
If we want to have a pure-DP privacy guarantee, then we can repeat the analysis above with parameter setup as
\[
\begin{cases}
\gamma = \frac{\log(20/\beta)}{\log(m)},\\
\tau = 5m^2, \\
\eps' = \frac{\eps}{\gamma + k}, \\
k' = k + 7\frac{\log(1/\beta)}{\log(m)}, \\
d = \frac{10\Delta}{\eps'} \log(m).
\end{cases}
\]
Fixing the privacy parameters $\eps$ and $k,m$, the tradeoff between error and confidence is 
\[
d\le O\left( \Delta \frac{k\log(m) + \log(1/\beta)}{\eps} \right).
\]
\end{remark}

\subsection{Applications}

In this section, we apply Algorithm~\ref{algo:repetitive-SVT} to improve the private multiplicative weight update algorithm and the sample complexity of adaptive data analysis.

\subsubsection{Private Multiplicative Weight Update}

It is instructive to review how the private multiplicative weight update algorithm \cite{DBLP:conf/focs/HardtR10} works. Roughly speaking, the private MWU algorithm maintains a publicly-known ``approximation'' distribution $h:\calX\to [0,1]$ to the uniform distribution over the private input $D$. Given a query $f_i:\calX\to [0,1]$, the algorithm uses $\Ex_{x\sim h}[f_i(x)]$ as a ``guess'' for $\Ex_{x\sim D}[f_i(x)]$. Then it defines a function $g_i(D) = \left| \Ex_{x\sim D}[f_i(x)] - \Ex_{x\sim h}[f_i(x)] \right|$, and uses SVT to privately check if $g_i(D) \lessapprox \alpha$, where $\alpha$ is the desired accuracy bound. Querying SVT yields two possible outcomes:
\begin{itemize}
    \item The SVT algorithm returns that $g_i(D) \lessapprox \alpha$. In this case, the algorithm answers this query with $\Ex_{x\sim h}[f_i(x)]$, without paying the privacy loss. 
    \item The SVT algorithm returns that $g_i(D) \gtrapprox \alpha$. Then, the algorithm answers this query by releasing $\Ex_{x\sim D}[f_i(x)] + \Lap(1/n\eps')$. After that, the algorithm also updates the current distribution $h$ according to the MWU rule. In this case, we call the $i$-th round an ``update round''.
\end{itemize}

It is shown in \cite{DBLP:conf/focs/HardtR10} that the number of update rounds is no more than $k = \frac{\log|\calX|}{\alpha^2}$. Therefore, if we desire the final algorithm to be $(\eps,\delta)$-DP, we can afford to use $(\eps',0)$-DP implementation of the sparse vector technique and the Laplace mechanism with $\eps' = \frac{\eps}{\sqrt{k\log(1/\delta)}}$. The standard calculation shows that, the magnitudes of Laplace noises used in the algorithm are all bounded by $O(\frac{\log(m/\beta)}{\eps'n})$ with probability $1-\beta$. Therefore, the best accuracy bound we can hope for is
\[
\alpha \approx \Theta\left(\frac{\log(m/\beta)}{\eps'n}\right)  = \Theta\left(\frac{\sqrt{k\log(1/\delta)}\log(m/\beta)}{\eps n}\right) = \Theta\left(\frac{\sqrt{\log|\calX|\log(1/\delta)}\log(m/\beta)}{\alpha\eps n}\right).
\]

This completes the proof sketch for Theorem~\ref{theo:MWU}. For reader's convenience, we recall the statement below.

\begin{reminder}{Theorem~\ref{theo:MWU}}
For every $\eps,\delta\in (0,1)$, there is an $(\eps,\delta)$-DP algorithm that answers $m$ adaptively chosen linear queries with the following utility guarantee. Let $(\tilde{f}_i)_{i\in [m]}$ denote the outputs of $\calA$. For every $\beta \in (0,1)$, with probability $1-\beta$ we have that $\max_{i\in [m]} \{ |\tilde{f}_i - \Ex_{j\in [n]}[f_i(x_j)]|\} \le\alpha_{\beta}$, where
\[
{\alpha_{\beta}}^2 \le O\left( \frac{\sqrt{\log|\calX|\log(1/\delta)} \log(m/\beta)}{n\eps} \right).
\]
\end{reminder}

\paragraph*{Our improvement.} As we have seen in the last subsection, using our improved SVT (Algorithm~\ref{algo:repetitive-SVT}), we can improve the confidence-accuracy tradeoff of SVT to
\[
\alpha \le O\left(  \frac{\sqrt{\log|\calX|\log(1/\delta)}\log(m)}{n\eps\alpha}+ \frac{\log(1/\beta)}{n \eps} \right).
\]
However, to yield the improved private MWU algorithm, there is one remaining issue. That is, during each update round, we need to use the Laplace mechanism to come up with an estimation of $\Ex_{x\sim D}[f_i(x)]$. With a small-but-noticeable probability, the Laplace noise $\Lap(1/n\eps')$ can be very large so that the estimation given by the Laplace mechanism is erroneous. We can resolve the issue by utilizing Algorithm~\ref{algo:repetitive-SVT}. Namely, after getting the estimation of $\Ex_{x\sim D}[f_i(x)]$ from the Laplace mechanism, we ask the repetitive SVT to test if the estimation is accurate. If it turns out to be inaccurate, we reject the current estimation and request a new one from the Laplace mechanism. Although this process consumes privacy budget, we can nevertheless show that the number of such ``rejections'' is small with overwhelmingly high probability, using an argument similar as the proof of Lemma~\ref{lemma:SVT-utility}. 

Another point worth mentioning is that the privacy analysis for Algorithm~\ref{algo:repetitive-SVT} is done for the whole execution. Therefore, when we use it as a subroutine in the private MWU algorithm, we are \emph{concurrently} running it with the Laplace mechanism. To argue the privacy of the algorithm, we need to use the concurrent composition theorems of DP (see, e.g., \cite{VadhanW21,VadhanZ22-concurrent,Lyu22-concurrent}).

To summarize, we have shown how we can arrive at the following improved private MWU algorithm.

\begin{reminder}{Theorem~\ref{theo:MWU-new}}
There is an absolute constant $C>0$ such that the following is true. For every $\eps,\delta\in (0,1)$, $m\in \mathbb{N}$ and $\beta \ge 2^{-m}$, there exists 
\[
\alpha \le C\cdot \left( \frac{\sqrt{\log|\calX|\log(1/\delta)} \log(m)}{n\eps\cdot \alpha} + \frac{\log(1/\beta)}{n\eps} \right)
\]
and an $(\eps,\delta)$-DP algorithm that answers $m$ adaptively chosen linear queries such that, with probability $1-\beta$, we have $\max_{i\in [m]} \{ |\tilde{f}_i - \Ex_{j\in [n]}[f_i(x_j)]|\}\le \alpha$.
\end{reminder}

\subsubsection{Adaptive Data Analysis}

The generalization property of DP \cite{DFHPTRR15,DBLP:conf/stoc/BassilyNSSSU16} provides a connection between privacy-preserving algorithms and algorithms for adaptive data analysis. In particular, an $(\eps,\delta)$-DP algorithm that answers $m$ linear queries with additive error $\alpha$ and confidence $1-\beta$, implies an $(\alpha',\beta',m)$-accurate algorithm for adaptive data analysis with linear queries, where $\alpha' = O(\alpha+\eps)$ and $\beta' = O(\beta + \delta/\eps)$.

Applying this generic connection to Theorem~\ref{theo:MWU} proves Theorem~\ref{theo:adaptive-analysis-known}. We now apply the connection with Theorem~\ref{theo:MWU-new} to prove Theorem~\ref{theo:adaptive-analysis-improved}, restated below.

\begin{reminder}{Theorem~\ref{theo:adaptive-analysis-improved}}
For every finite universe $\calX$, every $m\ge 1$ and $\alpha\in (0, 1), \beta \ge 2^{-m}$, there is an $(\alpha, \beta, m)$-accurate adaptive data analysis algorithm with sample complexity
\[
n \le O\left( \min\left( \frac{\sqrt{\log|\calX|\log (1/\beta)}}{\alpha^3}, \frac{\log|X|}{\alpha^4} \right) \log m + \frac{\log(1/\beta)}{\alpha^2}\right).
\]
\end{reminder}

Suppose we aim for an $(\alpha,\beta,m)$-accurate algorithm. Letting $k = \frac{\log|\calX|}{\alpha^2}$, Theorem~\ref{theo:MWU-new} and the generalization property imply that this is possible as long as the inequality below is satisfied
\[
\alpha \ge C'\cdot \left( \frac{\sqrt{\log|\calX|\log(1/\delta)} \log(m)}{n\eps\cdot \alpha} + \frac{\log(1/\beta)}{n\eps} \right)
\]
for some absolute constant $C' > 0$. Solving for $n$ yields that
\[
n\ge C' \left( \frac{\sqrt{k\log (1/\beta)}}{\alpha^2}\log m + \frac{\log(1/\beta)}{\alpha^2}\right).
\]
Therefore, we can choose $n$ to be as small as
\[
O\left( \frac{\sqrt{k\log (1/\beta)}}{\alpha^2}\log m + \frac{\log(1/\beta)}{\alpha^2}\right).
\]
In Theorem~\ref{theo:adaptive-analysis-improved}, we also claimed an incomparable bound
\[
n\le O\left( \frac{k}{\alpha^2}\log m + \frac{\log(1/\beta)}{\alpha^2}\right).
\]
This can achieved using a pure-DP version of Theorem~\ref{theo:MWU-new}, which can be derived from the pure-DP version of Algorithm~\ref{algo:repetitive-SVT} (see Remark~\ref{remark:SVT-pure-DP}).

\begin{small}
\paragraph*{Acknowledgements}
Edith Cohen is partially supported by Israel Science Foundation (grant no. 1595/19).
Uri Stemmer is partially supported by the Israel Science Foundation (grant 1871/19) and by Len Blavatnik and the Blavatnik Family foundation.

We are grateful to Thomas Steinke for helpful discussions and comments on an early version of the paper.
\end{small}

\bibliographystyle{alpha}
\bibliography{main}

\appendix

\section{Missing Proofs}\label{appendix:proof-privacy}

\subsection{Missing Proofs in Section~\ref{sec:metaprivacyproofs}}

In this section, we show the proof of Lemmas~\ref{lemma:max-divergence-one-sided} and \ref{lemma:divergence-one-sided}.
\subsubsection{Preliminaries}

Recall we have defined the $\alpha$-order R\'enyi divergence of $P$ from $Q$ as
\[
D_{\alpha}(P \| Q) := \frac{1}{\alpha - 1} \log \left( \sum_{x\in X} P(x) \left( \frac{P(x)}{Q(x)} \right)^{\alpha - 1} \right).
\]
The max-divergence is defined as
\[
D_{\infty}(P \| Q) := \log(\max_{x} \{P(x)/Q(x)\}).
\]

We need a composition property of R\'enyi divergence. Suppose $P,Q$ are distributions on $(x,y)\in X\times Y$. For each $x\in X$, let $P_x$ denote the conditional distribution of $y\in (x,y)\sim P$ conditioning on $x$. Also let $P^{X}$ denote the marginal distribution of $P$ on $X$. Also define similar pieces of notation for $Q$. Then, we have
\[
D_{\alpha}(P^X\| Q^X) + \inf_{x\in X} \{ D_{\alpha}(P_x \| Q_x) \} \le D_{\alpha}(P\|Q ) \le D_{\alpha}(P^X\| Q^X) + \sup_{x\in X} \{ D_{\alpha}(P_x \| Q_x) \}.
\]

\subsubsection{Proofs of Lemmas~\ref{lemma:max-divergence-one-sided} and ~\ref{lemma:divergence-one-sided}}

We will analyze the $0$-favored coin-flipping mechanism. For readers' convenience, we restate it below.

\begin{algorithm2e}[H]
\LinesNumbered
    \caption{(Restate of Algorithm~\ref{algo:coin-flip-asymmetric}) The $0$-favored coin-flipping mechanism.}
    \label{algo:coin-flip-asymmetric-restated}
    \DontPrintSemicolon
    \KwIn{
           An input bit $b\in \{0,1\}$, a privacy parameter $\eps \in (0,1)$, an integer $k\ge 1$.
    }
    \SetKwProg{Init}{initialize}{:}{}
    \Init{}{
        $c\gets 0$ \;
    }
    \While{$c<k$} {
        Receive a query $(p,q)\in [0,1]^2$, promised that $0 \le q\le p\le e^{\eps} q$ and $(1-p)\le e^{\eps} (1-q)$ \;
        \If{$b=0$}{
            Sample $r\sim \Ber(p)$\;
        }
        \Else{
            Sample $r\sim \Ber(q)$\;
        }
        \If{$r=1$}{
            $c\gets c + 1$ \;
        }
        Output $r$ \;
    }
\end{algorithm2e}

\begin{reminder}{Lemma~\ref{lemma:max-divergence-one-sided}}
Consider Algorithm~\ref{algo:coin-flip-asymmetric-restated}. Fix $\calA$ to be an arbitrary adversary interacting with Algorithm~\ref{algo:coin-flip-asymmetric-restated}. Let $P,Q$ denote the distributions of the interaction between $\calA$ and Algorithm~\ref{algo:coin-flip-asymmetric-restated} when the sensitive bit is $0$ or $1$, respectively. Then, for any $\alpha \in (1, \infty)$, we have $D_{\infty}(P\|Q) \le k\eps$.
\end{reminder}

\begin{proof}
We first consider the case that $k = 1$ and that the adversary is deterministic. In this case, note that the algorithm always outputs $0$ before a final round, in which it outputs a $1$ and halts. Therefore, any deterministic adversary strategy can be described by a list of probability pairs $(p_1,q_1),(p_2,q_2),\dots $. To interact with the algorithm, the adversary just keeps sending the next pair in the list, until it gets a $1$-response from the algorithm. 

Hence, the distribution $P$ and $Q$ can be described by two distributions on positive integers. Where $P = i$ denotes the outcome that the algorithm halts after the $i$-th round when it holds $b=0$. It is easy to see that
\[
\Pr[P = i] = \prod_{j < i} (1-p_j) p_i,
\]
Similarly, $Q = i$ denotes the outcome that the algorithm halts after the $i$-th round when it holds $b = 1$. We have
\[
\Pr[Q = i] = \prod_{j < i} (1 - q_j) q_i.
\]
One can then write 
\[
D_{\infty}(P\|Q) = \log\left( \sup_{i\ge 1} \left\{ \frac{\Pr[P = i]}{\Pr[Q = i]} \right\} \right).
\]
For every $i\ge 1$, it is easy to see
\[
\Pr[P = i] = \sum_{j < i} (1-p_j) p_i \le \prod_{j < i} (1-q_j) p_i \le e^\eps \sum_{j < i} (1-q_j) q_i = e^{\eps} \Pr[Q = i].
\]
This completes the proof for the base case. 

For the case that the adversary is randomized, we argue as follows. For any possible random coins of the adversary $r$, let $P_r,Q_r$ denote the distributions conditioning on $r$. We have
\[
D_{\infty}(P_r \| Q_r) \le \eps,
\]
because the adversary becomes deterministic conditioning on $r$. Hence, we have
\[
D_{\infty}(P \| Q) \le \sup_{r} \left\{ D_{\infty}(P_r \| Q_r) \right\} \le \eps
\]
as desired.

For the case of $k = c > 1$, we can think of the interaction as a concatenation of at most $c$ executions of the algorithm with $k = 1$ each. The proposition follows from the composition property of max-divergence.
\end{proof}

\begin{reminder}{Lemma~\ref{lemma:divergence-one-sided}}
Consider the same setup as in Lemma~\ref{lemma:max-divergence-one-sided}. For any $\alpha \in (1, \infty)$, we have $D_{\alpha}(P\|Q) \le 3k \alpha  \eps^2$.
\end{reminder}

\begin{proof}
We assume that $\alpha \eps < \frac{1}{3}$, as otherwise the conclusion follows easily from Lemma~\ref{lemma:max-divergence-one-sided}.

Again we start by considering the case that $k = 1$ and that the adversary is deterministic. We inherit the notation $(p_1,q_1),(p_2,q_2),\dots $ from the proof of Lemma~\ref{lemma:max-divergence-one-sided}. Recall
\[
\Pr[P = i] = \prod_{j < i} p_j  (1-p_i),
\]
and
\[
\Pr[Q = i] = \prod_{j < i} q_j  (1 - q_i).
\]
Then, we have
\begin{align}
e^{(\alpha-1)D_{\alpha}(P\| Q)} = \sum_{i=1}^{\infty} \Pr[P = i] \left( \frac{\Pr[P = i]}{\Pr[Q = i]}\right)^{\alpha - 1}. \label{eq:divergence-main}
\end{align}

We claim the following.
\begin{claim}\label{claim:induction-divergence}
Assume $\alpha \eps < 1/3$. For every $m \ge 1$, it holds that
\begin{align}
\sum_{i=1}^{m} \Pr[P = i] \left( \frac{\Pr[P = i]}{\Pr[Q = i]}\right)^{\alpha - 1} +\Pr[P > m] \left( \frac{\Pr[P > m]}{\Pr[Q > m]} \right)^{\alpha - 1}  \le 1 + 3\alpha(\alpha - 1) \eps^2.\label{eq:claim-induction}
\end{align}
\end{claim}

We assume Claim~\ref{claim:induction-divergence} for now and quickly finish the proof of Lemma~\ref{lemma:divergence-one-sided}. Given Claim~\ref{claim:induction-divergence}, take the limit $m \to \infty$ and plug it into \eqref{eq:divergence-main}:
\[
e^{(\alpha-1)D_{\alpha}(P\| Q)} = \sum_{i=1}^{\infty} \Pr[P = i] \left( \frac{\Pr[P = i]}{\Pr[Q = i]}\right)^{\alpha - 1} \le 1 + 3\alpha (\alpha - 1) \eps^2.
\]
Taking logarithm on both sides and using $\log(1+x)\le x$ show that $D_\alpha(P\| Q)\le 3\alpha \eps^2$. This completes the proof for the base case.

For the case that $k = c > 1$, we view the interaction as a sequential composition of $c$ executions of Algorithm~\ref{algo:coin-flip-asymmetric-restated} with $k=1$ each, and use the composition property of R\'enyi divergence. For the case of randomized adversaries, conditioning on the random coins of the adversary $r$, the conditional distributions of the interaction satisfy that 
\[
D_{\alpha}(P_r \| Q_r) \le 3\alpha k\eps^2.
\]
Then, we have
\[
D_{\alpha}(P\| Q) \le \sup_{r} \left\{ D_{\alpha}(P \| Q) \right\} \le 3\alpha k\eps^2,
\]
completing the proof\footnote{We remark that this first inequality will not be true if we replace the supremum operator over $r$ with expectation, as $D_{\alpha}$ is not necessarily convex (see \cite{DBLP:journals/tit/ErvenH14-renyi-divergence} for a counterexample).}.
\end{proof}

In the following, we prove Claim~\ref{claim:induction-divergence}. 

\begin{proofof}{Claim~\ref{claim:induction-divergence}}

We use induction on $m$. For the case of $m = 1$, we have
\begin{align}
& ~~~~ \Pr[P = 1]\left( \frac{\Pr[P = 1]}{\Pr[Q = 1]} \right)^{\alpha - 1}
+
\Pr[P > 1]\left( \frac{\Pr[P > 1]}{\Pr[Q > 1]} \right)^{\alpha - 1}
\notag \\
& = 
p_1 \left( \frac{p_1}{q_1} \right)^{\alpha - 1} + (1-p_1) \left( \frac{1-p_1}{1-q_1} \right)^{\alpha - 1}. \label{eq:base-case}
\end{align}

We need the following simple fact.

\begin{fact}\label{fact:bernoulli-divergence}
Assume $\alpha\eps\le 1/3$. Let $0\le p,q \le 1$ be such that
\[
\max(D_{\infty}(\Ber(p)\| \Ber(q)), D_{\infty}(\Ber(q)\|\Ber(p))) \le \eps.
\]
Then,
\[
e^{(\alpha - 1)D_{\alpha}(\Ber(p)\|\Ber(q))} \le 1 + \alpha (\alpha - 1)\eps^2.
\]
\end{fact}

We present a simple proof of Fact~\ref{fact:bernoulli-divergence} by Steinke \cite{Ste-personal}.
\begin{proof}
The promise on $p,q$ implies that $D_{\alpha}(\Ber(p)\|\Ber(q))\le \frac{\alpha \eps^2}{2}$ (see, e.g., \cite[Proposition 3.3]{DBLP:conf/tcc/BunS16}). Then, using the inequality $e^x\le 1+2x$ for $x\in (0,1.25)$ gives us that
\[
e^{(\alpha - 1)D_{\alpha}(\Ber(p)\|\Ber(q))} \le e^{\frac{\eps^2\alpha(\alpha-1)}{2}} \le 1 + \alpha (\alpha-1)\eps^2,
\]
as desired.
\end{proof}

By Fact~\ref{fact:bernoulli-divergence}, we have
\[
p_1 \left( \frac{p_1}{q_1} \right)^{\alpha - 1} + (1-p_1) \left( \frac{1-p_1}{1-q_1} \right)^{\alpha - 1} \le 1 + (\alpha - 1)\alpha \eps^2.
\]
This verifies the base case for $m=1$. Now suppose Claim~\ref{claim:induction-divergence} is true for $m - 1 \ge 1$, we prove for the case of $m$. We have
\begin{align}
&~~~~\sum_{i=1}^{m} \Pr[P = i] \left( \frac{\Pr[P = i]}{\Pr[Q = i]}\right)^{\alpha - 1} +\Pr[P > m] \left( \frac{\Pr[P > m]}{\Pr[Q > m]} \right)^{\alpha - 1} \notag \\
&=
\Pr[P = 1]\left( \frac{\Pr[P = 1]}{\Pr[Q = 1]} \right)^{\alpha - 1} + 
\Pr[P > 1] \left( \frac{\Pr[P > 1]}{\Pr[Q > 1]} \right)^{\alpha - 1} \times \notag  \\
& ~~~~ \left( \sum_{i=2}^{m} \Pr[P = i | P > 1] \left( \frac{\Pr[P = i | P > 1]}{\Pr[Q = i | Q > 1]}\right)^{\alpha - 1} + \Pr[P > m | P > 1] \left( \frac{\Pr[P > m| P > 1]}{\Pr[Q > m | Q > 1]} \right)^{\alpha - 1} \right). \label{eq:claim-step-2} 
\end{align}
We look into the bracket on the third line. Conditioning on $P > 1$ (and $Q > 1$), the expression inside the bracket describes a $(m-1)$-round interaction between the adversary and the algorithm, where the probability pairs are given by $(p_2,q_2),\dots, (p_m, q_m)$. Therefore, by induction hypothesis, it is subject to the bound given by \eqref{eq:claim-induction}, which allows to bound \eqref{eq:claim-step-2} by
\begin{align}
&~~~~\sum_{i=1}^{m} \Pr[P = i] \left( \frac{\Pr[P = i]}{\Pr[Q = i]}\right)^{\alpha - 1} +\Pr[P > m] \left( \frac{\Pr[P > m]}{\Pr[Q > m]} \right)^{\alpha - 1} \notag  \\
&\le 
\Pr[P = 1]\left( \frac{\Pr[P = 1]}{\Pr[Q = 1]} \right)^{\alpha - 1} + 
\Pr[P > 1] \left( \frac{\Pr[P > 1]}{\Pr[Q > 1]} \right)^{\alpha - 1} \cdot \left( 1 + 3\alpha (\alpha - 1) \eps^2 \right). \label{eq:claim-step-3} 
\end{align}
Now, depending on whether $p_1 = \Pr[P = 1]$ is large or small, we consider two cases.

\medskip\noindent\textbf{Case 1: $p_1 \ge \frac{2}{3}$.} In this case, we use $\Pr[P > 1] = 1-p_1 \le \frac{1}{3}$, $\Pr[P > 1] \le \Pr[Q > 1]$, and Fact~\ref{fact:bernoulli-divergence} to bound \eqref{eq:claim-step-3} as
\[
\begin{aligned}
&~~~~ \Pr[P = 1]\left( \frac{\Pr[P = 1]}{\Pr[Q = 1]} \right)^{\alpha - 1} + 
\Pr[P > 1] \left( \frac{\Pr[P > 1]}{\Pr[Q > 1]} \right)^{\alpha - 1} \cdot \left( 1 + 3\alpha (\alpha - 1) \eps^2 \right) \\
&= p_1 \left( \frac{p_1}{q_1} \right)^{\alpha - 1} + (1-p_1) \left( \frac{1-p_1}{1-q_1} \right)^{\alpha - 1} + (1-p_1) \left( \frac{1-p_1}{1-q_1} \right)^{\alpha - 1} \cdot 3\alpha (\alpha - 1) \eps^2 \\
&\le 1 + \alpha(\alpha - 1)\eps^2 + (1-p_1) 3 \alpha (\alpha - 1) \eps^2 \\
&\le 1 + 3 \alpha(\alpha - 1) \eps^2.
\end{aligned}
\]

\medskip\noindent\textbf{Case 2: $p_1 < \frac{2}{3}$.} In this case, write $\lambda = \log(p_1/q_1) \in [0, \eps]$. We bound \eqref{eq:claim-step-3} as
\begin{align}
&~~~~ \Pr[P = 1]\left( \frac{\Pr[P = 1]}{\Pr[Q = 1]} \right)^{\alpha - 1} + 
\Pr[P > 1] \left( \frac{\Pr[P > 1]}{\Pr[Q > 1]} \right)^{\alpha - 1} \cdot \left( 1 + 3\alpha (\alpha - 1) \eps^2 \right) \notag \\
&\le p_1  e^{\lambda (\alpha - 1)} + (1-p_1) \left(1 - \frac{p_1 - q_1}{1 - q_1} \right)^{\alpha - 1} + (1-p_1) \cdot 3\alpha (\alpha - 1) \eps^2 \notag \\
&\le p_1 e^{\lambda (\alpha - 1)} + (1-p_1) e^{-\frac{p_1-q_1}{1 - q_1}(\alpha - 1)} + (1-p_1) \cdot 3\alpha (\alpha - 1) \eps^2, \label{eq:claim-step-4}
\end{align}
where we used $\Pr[P>1] \le \Pr[Q > 1]$ and $(1-x)^y \le e^{-xy}$ in the second and third line, respectively. 

Then we re-write \eqref{eq:claim-step-4} as
\begin{align}
p_1 (e^{\lambda (\alpha - 1)} - 1) + (1-p_1) (e^{-\frac{p_1-q_1}{1 - q_1}(\alpha - 1)} - 1) + 1 + (1-p_1) \cdot 3\alpha (\alpha - 1) \eps^2. \label{eq:claim-step-5}
\end{align}
Now, it is clear that to bound \eqref{eq:claim-step-5} by $3\alpha (\alpha - 1) \eps^2$, it suffices to show that 
\begin{align}
p_1 (e^{\lambda (\alpha - 1)} - 1) + (1-p_1) (e^{-\frac{p_1-q_1}{1 - q_1}(\alpha - 1)} - 1) \le p_1\cdot 3 \alpha(\alpha - 1) \eps^2. \label{eq:claim-step-6}
\end{align}
We use the basic inequality $e^{x} \le 1 + x + x^2$, which is valid for $x\in (-\infty, 1.75]$, to bound \eqref{eq:claim-step-6} as
\begin{align}
&~~~~ p_1 (e^{\lambda (\alpha - 1)} - 1) + (1-p_1) (e^{-\frac{p_1-q_1}{1 - q_1}(\alpha - 1)} - 1) \notag \\
&\le p_1 ( \lambda(\alpha - 1) + \lambda^2 (\alpha - 1)^2 ) + (1-p_1) \left( -\frac{p_1-q_1}{1-q_1}(\alpha - 1) + \frac{(p_1-q_1)^2}{(1-q_1)^2}(\alpha - 1)^2 \right). \label{eq:claim-step-7}
\end{align}
Recall $q_1 = e^{-\lambda} p_1$ and $0\le \lambda \le \eps \le 1/3$, which implies that
\[
p_1 - q_1 = p_1(1 - e^{-\lambda}) \in \big( p_1 (\lambda - \lambda^2), p_1 \lambda \big).
\]
We proceed to bound \eqref{eq:claim-step-7} as
\begin{align}
& ~~~~  p_1 ( \lambda(\alpha - 1) + \lambda^2 (\alpha - 1)^2 ) + (1-p_1) \left( -\frac{p_1-q_1}{1-q_1}(\alpha - 1) + \frac{(p_1-q_1)^2}{(1-q_1)^2}(\alpha - 1)^2 \right) \notag \\
& = p_1 \lambda(\alpha -1) + p_1 \lambda^2 (\alpha-1)^2 - \frac{1-p_1}{1-q_1} (p_1 - q_1)(\alpha - 1) + \frac{(1-p_1)}{(1-q_1)^2}(p_1-q_1)^2(\alpha - 1)^2. \label{eq:claim-step-8}
\end{align}
Observe that
\[
\begin{aligned}
&~~~~ p_1 \lambda(\alpha -1) -  \frac{1-p_1}{1-q_1} (p_1 - q_1)(\alpha - 1) \\
&\le p_1 \lambda(\alpha -1) -  (1 - \frac{p_1 - q_1}{1-q_1}) (p_1 - q_1)(\alpha - 1) \\
&\le p_1 \lambda(\alpha - 1) -  (p_1 - q_1)(\alpha - 1) + \frac{p_1 - q_1}{1-q_1} (p_1 - q_1)(\alpha - 1) \\
&\le \lambda^2 p_1 (\alpha - 1) + \frac{p_1^2}{1-q_1} \lambda^2  (\alpha - 1) & \text{($p_1-q_1\in (p_1\lambda - p_1\lambda^2, p_1\lambda)$)}\\
&\le p_1 \lambda^2 (\alpha - 1) + 2 p_1 \lambda^2  (\alpha - 1), & \text{($\frac{p_1}{1-q_1}\le 2$)}
\end{aligned}
\]
and
\[
\begin{aligned}
&~~~~ p_1 \lambda^2 (\alpha-1)^2 + \frac{(1-p_1)}{(1-q_1)^2}(p_1-q_1)^2(\alpha - 1)^2 \\
&\le p_1 \lambda^2 (\alpha - 1)^2 + 2 p_1 \lambda^2 (\alpha - 1)^2. & \text{($\frac{p_1}{1-q_1}\le 2$ and $\frac{1-p_1}{1-q_1}\le 1$)}
\end{aligned}
\]
Plug them back in \eqref{eq:claim-step-8}. We conclude that
\begin{align}
& ~~~~ p_1 \lambda(\alpha -1) + p_1 \lambda^2 (\alpha-1)^2 - \frac{1-p_1}{1-q_1} (p_1 - q_1)(\alpha - 1) + \frac{(1-p_1)}{(1-q_1)^2}(p_1-q_1)^2(\alpha - 1)^2 \notag \\
&\le 3 p_1 \lambda^2 (\alpha - 1) + 3 p_1 \lambda^2 (\alpha - 1)^2 \notag \\
&\le 3p_1 \lambda^2 (\alpha - 1) \alpha. \label{eq:cliam-step-9}
\end{align}
Hence, by a series of deductions (\eqref{eq:claim-step-7}, \eqref{eq:claim-step-8}, and \eqref{eq:cliam-step-9}), we have finally justified \eqref{eq:claim-step-6}, which completes the proof.
\end{proofof}

\end{document}